\documentclass[runningheads,a4paper]{llncs}
\usepackage{a4wide}

\usepackage[latin1]{inputenc}
\usepackage{graphicx}
\usepackage{amsfonts}
\usepackage{amssymb}
\usepackage{amsmath}
\usepackage{framed}
\usepackage{latexsym}
\usepackage{breakcites}
\usepackage{color}
\usepackage{url}
\usepackage{tikz}
\usepackage{algorithm}
\usepackage{caption}
\usepackage[noend]{algpseudocode}
\usepackage{enumerate}
\usepackage{stmaryrd}
\usepackage{float}

\usepackage{thm-restate}

\makeatletter
\newcounter{restate}

\makeatother

\newcommand{\cA}{\mathcal{A}}
\newcommand{\cB}{\mathcal{B}}

\newcommand{\cD}{\mathcal{D}}
\newcommand{\cE}{\mathcal{E}}
\newcommand{\cF}{\mathcal{F}}
\newcommand{\cG}{\mathcal{G}}
\newcommand{\cH}{\mathcal{H}}

\newcommand{\cS}{\mathcal{S}}

\newcommand{\cU}{\mathcal{U}}
\newcommand{\cV}{\mathcal{V}}

\newcommand{\av}{{\mathbf{a}}}
\newcommand{\cv}{{\mathbf{c}}}
\newcommand{\ev}{{\mathbf{e}}}
\newcommand{\mv}{{\mathbf{m}}}
\newcommand{\rv}{{\mathbf{r}}}
\newcommand{\sv}{{\mathbf{s}}}
\newcommand{\uv}{{\mathbf{u}}}
\newcommand{\vv}{{\mathbf{v}}}
\newcommand{\xv}{{\mathbf{x}}}
\newcommand{\yv}{{\mathbf{y}}}
\newcommand{\zv}{{\mathbf{z}}}
\newcommand{\Am}{{\mathbf{A}}}
\newcommand{\Bm}{{\mathbf{B}}}

\newcommand{\Hm}{{\mathbf{H}}}

\newcommand{\Pm}{{\mathbf{P}}}
\newcommand{\Sm}{{\mathbf{S}}}
\newcommand{\Um}{{\mathbf{U}}}

\newcommand{\hash}{{\mathcal{H}}}
\newcommand{\Cc}{{\mathcal C}}
\newcommand{\Dc}{{\mathcal D}}
\newcommand{\Fc}{{\mathcal F}}

\newcommand{\Hpub}{{\Hm_{\textup{pub}}}}
\newcommand{\Hsec}{{\Hm_{\textup{sec}}}}
\newcommand{\Csec}{\Cc_{\textup{sec}}}

\newcommand{\F}{\mathbb{F}}

\newcommand{\Unif}{\hookleftarrow}

\newcommand{\eqdef}{\mathop{=}\limits^{\triangle}}

\DeclareMathOperator*{\Ver}{\text{\tt{Vrfy}}}
\DeclareMathOperator*{\Sgn}{\text{\tt{Sgn}}}
\DeclareMathOperator*{\Gen}{\text{\tt{Gen}}}
\newcommand{\SD}{\ensuremath{\mathrm{SD}}}
\newcommand{\QDOOM}{\ensuremath{\mathrm{DOOM}_{\infty}}}
\newcommand{\DOOM}{\ensuremath{\mathrm{DOOM}}}
\newcommand{\QROM}{\ensuremath{\mathrm{QROM}}}

\newcommand{\DST}{{\ensuremath{ \mathrm{SURF} }}}

\newcommand{\Sgnsk}{\ensuremath{\mathtt{Sgn}^{\mathrm{sk}}}}
\newcommand{\Vrfypk}{\ensuremath{\mathtt{Vrfy}^{\mathrm{pk}}}}

\newcommand{\wt}[1]{|#1|}

\newcommand{\qhash}{q_{\textup{hash}}}
\newcommand{\qsig}{q_{\textup{sign}}}

\newcommand{\Initialize}{\mathsf{Initialize}}
\newcommand{\Hash}{\mathsf{Hash}}
\newcommand{\Sign}{\mathsf{Sign}}
\newcommand{\Finalize}{\mathsf{Finalize}}
\newcommand{\Dpub}{\Dc_{\textup{pub}}}
\newcommand{\Drand}{\Dc_{\textup{rand}}}
\usepackage{booktabs}% http://ctan.org/pkg/booktabs

\newcommand*{\xdash}[1][3em]{\rule[0.5ex]{#1}{0.55pt}}
\newcommand{\horizontall}{\begin{center}\vspace*{-0.1cm}\xdash[35em]\vspace*{-0.1cm}\end{center}}

\newcommand{\COMMENT}[1]{}

\newcommand{\ket}[1]{|#1\rangle}

\def\01{\{0,1\}}
\def\01{\{0,1\}}

\newcommand{\eps}{\varepsilon}

\newcommand{\zo}{\{0,1\}}

\begin{document}

\title{A tight security reduction in the quantum random oracle model for code-based signature schemes}
\titlerunning{A tight security reduction in the QROM for code-based signature schemes}
%\runningtitle{zzz}
\author{Andr\'{e} Chailloux\inst{2} \and Thomas Debris-Alazard\inst{1,2}
}

\institute{Sorbonne Universit\'{e}s, UPMC Univ Paris 06\and Inria, Paris\\
	\email{\{andre.chailloux,thomas.debris\}@inria.fr}}

\maketitle

\begin{abstract}
	
	Quantum secure signature schemes have a lot of attention recently, in particular because of the NIST call to standardize quantum safe cryptography. However, only few signature schemes can have concrete quantum security because of technical difficulties associated with the Quantum Random Oracle Model (QROM). In this paper, we show that code-based signature schemes based on the full domain hash paradigm can behave very well in the QROM \textit{i.e.} that we can have tight security reductions. We also study quantum algorithms related to the underlying code-based assumption. Finally, we apply our reduction to a concrete example: the SURF signature scheme. We provide parameters for 128 bits of quantum security in the QROM and show that the obtained parameters are competitive compared to other similar quantum secure signature schemes.

\end{abstract}

\textbf{Keywords:} Quantum Random Oracle, Quantum-Safe Cryptography, Code-Based Cryptography, Signature Scheme, Quantum Decoding Algorithm

\section{Introduction}
Quantum computers are a potential big threat for many public-key cryptosystems. Every cryptographic application based on the hardness of factoring or the discrete logarithm (in a finite field or on elliptic curves) can be broken using Shor's quantum algorithm \cite{Sho94}. Even if the future existence of such quantum hardware is still in question, there is a rapidly increasing interest in building cryptosystems which are secure against quantum computers, field sometimes referred to as \emph{post-quantum} cryptography, or \emph{quantum-safe} cryptography. 

There are at least a couple of reasons why we should care and develop this line of research. First, sensitive data is now stored online and we want to guarantee long-term security. Indeed, we do not want current medical, political or other sensitive data to be stored today and decrypted in let's say 20 to 30 years from now, so cryptographic applications should try to find a way to prevent this kind of attacks. Quantum-safe cryptography would prevent a potential quantum computer to break today's schemes. Moreover, creating alternatives to RSA-based schemes could also be useful if another (non quantum) attack is found on factoring or the discrete log. While this doesn't seem to be the most probable, there has been for example big improvements on the discrete logarithm problem \cite{BGGM15} so such a scenario is not totally out of the question. Those concerns are strong and initiated a NIST call to standardize quantum-safe cryptography \cite{NIST16}, and also lead to industrial interest \cite{Goo16}. 

In this paper, we study (classical) signature schemes secure against quantum adversaries. Digital signature schemes allow to authenticate messages and documents and are a crucial element of many cryptographic applications such as software certification. There are several proposals for quantum-safe signatures based mostly on the hardness of lattice problems, such as BLISS \cite{DDLL13}, GPV \cite{EB14} or TESLA \cite{ABB+17}. There are also other quantum-safe assumptions that can be used such as the hardness of code-based problems, multivariate polynomial problems \cite{DS05} or the quantum security of hash functions \cite{BHH+15}. All of those quantum-secure signature schemes have different merits and imperfections. Some have good time and size parameters but use a very structured lattice-based assumption. Others have large  key and/or signature sizes and can have large running times.

Code based cryptography is among the oldest proposals for modern cryptography but suffered historically from the difficulty to construct a good signature scheme. The underlying computational assumption was actually one of the first proposed computational assumption \cite{Mce78} and still resists to known classical and quantum attacks. Until recently, there were very few proposals that were able to perform a code-based digital signature, the most notable being the CFS signature scheme \cite{CFS01}. However, a very recent proposal, the {\DST} signature scheme (see \cite{DST17}), presents competitive parameters (comparable to TESLA) but the security was shown only against classical adversaries.

Actually, most of signature schemes listed above - even though they use a quantum-safe computational assumption - can only prove security against classical adversaries. In fact, as of today only SPHINCS and {TESLA-2} have full security reductions which claim 128 bits of quantum security. This small amount of signature schemes comes from the difficulties to deal with the quantum random oracle model (QROM). In most of the security proofs used, we are in the random oracle model meaning that we use a hash function that behaves as a random function. A quantum attacker could still perform superposition attacks on this hash function and this creates many difficulties in the security reductions. There has already been a extensive amount of work to provide security reduction in the QROM\cite{BDF+11,Zha12}. However, most of them are not tight and there are significant losses in the parameters that can be used. TESLA recently managed to overcome those problems in the QROM while SPHINCS does not require the random oracle all together.

One of the most standard constructions for signature schemes is the Full Domain Hash (FDH) paradigm. In its most basic form, the idea is to use a trapdoor one-way function $f$, informally a function that can be efficiently inverted only with some secret key but that can be computed with the public key available. The signature of a message $\mv$ is a string $\xv$ such that $f(\xv) = \mathcal{H}(\mv)$ where $\mathcal{H}$ is a hash function, modeled in the ROM as a random function. Such a signature for $\mv$ can be done only by a signer which has access to the secret key. 

There are many constructions for signature schemes which use the FDH paradigm \cite{BR96,CFS01,BLS04}. Some of them can be proved secure even against quantum adversaries in the QROM, for example when the security reduction is history free \cite{BDF+11}. However, those reductions are, in many case, not tight. Indeed, one usually needs to reprogram the random oracle in the security proof and this is usually costly - especially in the quantum setting. This is one of the reasons why there are so few signature schemes with concrete quantum security parameters with a quantum reduction. However, as the NIST competition arrives, it becomes increasingly important to develop signature schemes, and associated security proofs, in order to provide fully quantum-safe cryptography.

\subsection*{Contributions}
In this work we study code-based signatures in the Full Domain Hash (FDH) paradigm. We will show under which conditions we can perform a quantum security reduction in the QROM for such schemes. While our work was strongly motivated by a recent construction of the {\DST} signature scheme, it can apply to different constructions, in particular to a different choice of codes and metrics like the rank metric.

We start from a family of error correcting codes $\mathcal{F}$ from which we can construct a trapdoor one way function $f$. The FDH paradigm then allows us to construct a signature scheme (for more details, see Section \ref{sec:signatureScheme}). We show the following results on this signature scheme:

\begin{enumerate}
	\item We show conditions on the code family $\mathcal{F}$ used such that the resulting signature scheme is secure against quantum adversaries in the QROM. Under these conditions, we present tight security reductions to the $\QDOOM$ problem \cite{JJ02,S11} (the Decode One Out of Many problem), which is an already used and studied variant of the standard syndrome decoding problem, where we have the choice between many words to decode instead of a single word (the $\infty$ subscript indicates that we do not limit this number).
	\item  We perform a complete analysis of quantum algorithms for the $\QDOOM$ problem, which can serve a reference for future work. The main idea here is to use the best known quantum algorithms for the $4$-Sum problem and reduce $\QDOOM$ to this problem.
	\item We apply our security reduction to a specific signature scheme: SURF (see \cite{DST17}). In this scheme, the family of codes $\mathcal{F}$ used there satisfies all the requirements of point 1. Thanks to our reduction, we can provide a full quantum security proof of this scheme. We get  concrete parameters for which we have 128 bits of security in the QROM. We also compare the parameters of SURF with the ones of other schemes and show that it is competitive as a quantum-safe signature scheme.
\end{enumerate}

In our first contribution, the security reduction, we actually manage to avoid most of the problems of the QROM. In particular, we do not reprogram the random oracle by injecting an instance of a hard problem. From a purely abstract way, this is done by considering in the FDH paradigm a one-way trapdoor function for which it is essentially as hard to find one out of many preimages. This already appeared implicitly in security reduction for FDH-like signature schemes but was handled with challenge injection and resulted in a non-tight security proof. More precisely, we will consider in this paper a one-way function $f$ such that, for a set of random and independent elements $\{\yv_{1},\cdots,\yv_{q}\}$ where we have the choice of the $q$ we consider, it stays hard to find $(\xv,i)$ such that $f(\xv) = \yv_{i}$. This \textit{One Out of Many problem} is clearly easier that the problem of inversion ($q=1$) but in this new paradigm we have seen a drastic advantage by considering it instead of performing instance injection in order to have a security proof in the \QROM. Quantum security proofs seem more natural and flexible with this approach and could be used outside of code-based cryptography.
%\newline
%
%It is the first time we have seen such a drastic advantage in considering a One Out of Many problem instead of performing instance injection.  \\

\subsubsection{Why is it that in code-based signatures, we can afford to work on a `One Out of Many' variant of Syndrome Decoding?}

$ \ $ \\

The most standard problem in code-based cryptography is the syndrome decoding (SD) problem:

\begin{restatable}{problem}{problemSD}[Syndrome Decoding - $\SD$]\label{prob:SDH}~\\
	\begin{tabular}{ll}
		Instance: & $\quad\Hm\in\F_{2}^{(n-k)\times n}$, $\sv\in\F_{2}^{n-k}$, $w$ integer \\
		Output: & $\quad\ev\in\F_{2}^n$ such that $\wt{\ev}=w$ and $\Hm\ev^T=\sv^{T}$
	\end{tabular}
\end{restatable}

We instead rely on the $\QDOOM$ problem:

\begin{restatable}{problem}{problemQDOOM}[\QDOOM]~\\
	\label{prob:qdoom}
	\begin{tabular}{ll}
		Instance: & $\quad\Hm \in \mathbb{F}_{2}^{(n-k) \times n}$ ; $\mathcal{H}$ a hash function in the QROM which takes its \\ 
		&\quad  values in $\mathbb{F}_{2}^{n-k}$ \\
		Output: &\quad $\ev \in \mathbb{F}_{2}^{n}$ of Hamming weight $w$, $\av \in \mathbb{F}_{2}^{*}$ such that, $\Hm\ev^{T} = \mathcal{H} (\av)^{T}$
	\end{tabular}
\end{restatable}

Here, we do not have a single input $\sv$ but we can generate as many inputs $\mathcal{H}(\av)$ as we want and we only need to solve the SD problem on one of the inputs. In the quantum setting, we even have access to a quantum oracle version of $\mathcal{H}$. It seems at first sight that this second problem is substantially easier than the first one. For example, when performing a brute force algorithm for SD, then this algorithm can be used to solve $\QDOOM$ $q$ times faster if we add $q$ queries to $\mathcal{H}$. However, the best classical and quantum algorithms for SD are much better than the brute force algorithm.
What actually happens is that the best known classical algorithms for $\QDOOM$ are not that much faster than those for SD. This running time difference decreases even more when looking at parameters used in the SURF signature scheme. Moreover, the quantum setting does not offer in the state-of-the-art a fully quadratic advantage compared to classical case for solving {\DOOM$_{\infty}$} as we will see below. This, combined with our tight security reduction, will allow us to give parameters for the SURF signature scheme for a quantum security of 128 which correspond to a classical security smaller than 256 bits.  \\

In our second contribution, we make the above explicit for quantum algorithms as well. The SD problem has been widely studied both classically and quantumly \cite{KT17}. There is - for most parameters - an algorithmic technique that does significantly better than others for this problem: the information set decoding technique first presented by Prange \cite{P62} and then improved several times \cite{S88,D91,MMT11,BJMM12,MO15}. 
Similarly, the best {\QDOOM} algorithms use the same method and the current state-of-the-art can be found in \cite{S11}.

The best asymptotic exponent among all those decoding techniques are \cite{MO15,BJMM12} for SD. However, algorithm \cite{MO15} is penalized by a big polynomial overhead which makes it more expensive that \cite{BJMM12}. It is why in the following table we will consider asymptotic exponents given by \cite{BJMM12}.  We give in Table \ref{table:classExp} classical exponents in base 2 of the Prange algorithm (which was the first algorithm proposed to solve syndrome decoding problem), \cite{BJMM12} and the state-of-the-art to solve {\DOOM}$_{\infty}$ (see \cite{S11}). We present the running times for $k = n/2$ and for two error weights $w$: namely $w \approx 0.11n$ which corresponds to the Gilbert-Varshamov weight and is the weight around which those problems are the hardest; and $w \approx 0.191n$ which corresponds to the weight used in the {\DST} signature scheme.

\begin{table}[H]
	\centering
	\begin{tabular}{|c||c|c|c|}\cline{2-4}
		\multicolumn{1}{c|}{} & \multicolumn{3}{|c|}{Classical asymptotic exponent in base $2$ (divided by $n$) } \\
		\hline
		$w/n$ & \qquad SD (Prange) \qquad \qquad  & \qquad SD (\cite{BJMM12}) \qquad \qquad  & $\QDOOM$ \cite{S11} \\
		\hline
		\hline 
		0.11 & 0.1199
		& 0.1000 & 0.0872 \\
		0.191 & 0.02029 & 0.01687 & 0.01654  \\
		\hline
	\end{tabular}
	\vspace{0.5cm}
	\caption{Asymptotic exponent for classically solving SD and {\DOOM}$_{\infty}$ for
		$k/n=0.5$}
	\label{table:classExp} 
\end{table}
\vspace*{-0.8cm}

The above table contains classical asymptotic exponent in base $2$ (divided by $n$). This means for example that the Prange algorithm for SD with $w = 0.11n$ runs in time $2^{0.1199n}$. \\

We extend in this paper the best {\QDOOM} algorithms to the quantum setting. We first present an overview of existing algorithms and we then show that the best known quantum algorithms for {\QDOOM} are very close, in complexity to the best known quantum algorithms for SD. Table \ref{table:quantExp} compares our algorithm to the current quantum knowledge for the same range of parameters. We will come back to these tables in \S\ref{sec:qdoom}. 
%Table \ref{table:classExp} gives some asymptotic exponent in the classical case while Table \ref{table:quantExp} compares our algorithm to the current quantum knowledge for the same range of parameters. We will come back to these tables in \S\ref{sec:qdoom}.  

\begin{table}[H]
	\centering
	\begin{tabular}{|c||c|c|c|}\cline{2-4}
		\multicolumn{1}{c}{}& \multicolumn{3}{|c|}{Quantum asymptotic exponent in base $2$ (divided by $n$)} \\ 
		\hline 
		$w/n$ &  SD (Prange)  & SD \cite{KT17}  &  \QDOOM (this work)  \\
		\hline
		\hline
		0.11 & 0.059958 & 0.058434
		& 0.056683
		\\
		0.191 & 0.010139 & 0.009218
		& 0.009159
		\\
		\hline
	\end{tabular}
	\vspace{0.3cm}
	\caption{Asymptotic exponent for quantumly solving SD and {\DOOM}$_{\infty}$ for
		$k/n=0.5$}
	\label{table:quantExp}
\end{table} 
\vspace*{-0.8cm}

As we can see, the best asymptotic exponents between the SD problem and the $\QDOOM$ problem are very close, especially for $w \approx 0,191$ which corresponds to the parameters of the {\DST} signature scheme. This allows us to greatly improve the security reduction in the QROM compared to the case where we would have used SD as a hard problem and performed challenge injection. \\

In our third contribution, we use the results presented above on the SURF signature scheme. As we said, there are very few signature schemes that claim quantum security. We present in table \ref{table:signatureparameters} security parameters for known quantum-safe (with a quantum security reduction) signature schemes. This data is taken from \cite{ABB+17}, where we added parameters for the {\DST} scheme obtained here.

\begin{table}[H]
	\centering
	\caption{Security parameters for signature schemes with quantum security claims}
	\label{table:signatureparameters}
	\begin{tabular}{|c|c|c|c|c|}
		\hline\noalign{\smallskip}
		Scheme & Quantum security  & Public key size  & Private key size & Signature size \\
		& (in bits) & (in kBytes) & (in kBytes) & (in kBytes) \\
		\noalign{\smallskip}
		\hline
		\noalign{\smallskip}
		SPHINCS & $128$ & $1$  & $1$ & $41$ \\
		GPV-poly & 59 & $55$ & $26$ & $32$ \\
		GPV & 59 & $27840$ & $12064$ & $30$ \\
		TESLA-2 & $128$ & $21799$ & $7700$ & $4$ \\
		SURF & $128$ & $5960$ & $3170$ & $1.7$ \\
		\hline
	\end{tabular}
\end{table}

We only presented here signature schemes for which quantum security is provided. There are many other signature that rely on a quantum-secure computational assumption but the full parameter analysis is not provided. We refer to \cite{ABB+17} for further details on this topic.

There is also another important metric that we do not discuss here: the running time of the different signature schemes. We did not add them here since both TESLA-2 and SURF do not have those available yet. Also the main contribution of our paper is to present an efficient security reduction, and not to compare in detail existing signature schemes.

\subsection*{Organisation of the paper}

After presenting some notations, we provide in Section \ref{sec:qqrom} a description of the quantum random oracle model. In Section \ref{sec:signatureScheme}, we present the general construction of code-based FDH signatures schemes and code-based problems. In Section \ref{sec:preliminaries}, we present some general preliminaries as well as security notions for signature schemes. In Section \ref{sec:securityProof}, we present the quantum security proof in the QROM. In Section \ref{sec:qdoom}, we study quantum algorithms for the $\QDOOM$ problem. In Section \ref{sec:qsurf} we apply our security reduction to the {\DST} signature scheme and show concrete parameters that achieve 128 bits of quantum security. Finally in Section \ref{sec:conclusion}, we perform a small discussion about the obtained results, and present directions for future research.

\subsection*{Notations}

We provide here some notation that will be used throughout the paper.  
Vectors will be written with  bold letters (such as $\ev$) and  uppercase bold letters are used to denote matrices (such as $\Hm$). Vectors are in row notation.
Let $\xv$ and $\yv$ be two vectors, we will write $(\xv|\yv)$ to denote their concatenation.
%We also denote for a subset $I$ of positions of the vector $\xv=(x_i)_{1 \leq i \leq n}$ by $\xv_I$ the vector whose components are those of $\xv$ which are indexed by $I$, i.e. 
%$$
%\xv_I = (x_i)_{i \in I}.
%$$ 
%We define the support of $\xv$ as
%$$
%\Sp(\xv) \eqdef \{ i \in \{1,\cdots,n \} \mbox{ such that } x_{i} \neq 0 \}
%$$
The Hamming weight of $\xv$ is denoted by 	
$|\xv|$.
By some abuse of notation, we will use the same notation 
to denote the size of a finite set: $|S|$ stands for the size of the finite set $S$.  
It will be clear from the context whether $|\xv|$ means the Hamming weight or the size of a finite set. 
%Note that 
%$$
%|\xv| = |\Sp(\xv)|.
%$$
The notation $x \eqdef y$ means that $x$ is defined to be equal to  $y$. We denote by $\mathbb{F}_{2}^{n}$ the set of binary vectors of length $n$ and 
$S_{w}$ is its subset of words of weight $w$. Let $S$ be a finite set, then $x \Unif S$ means 
that $x$ is assigned to be a random element chosen uniformly at random in $S$. For a distribution $\Dc$ we write $\xi \sim \Dc$ to indicate that the random variable $\xi$ is chosen according to $\Dc$.
%We denote the uniform distribution on $S_{w}$ by $\mathcal{U}_{w}$. 
%
%
%A binary linear code $\mathcal{C}$ of length $n$ and dimension $k$ is a subspace of $\mathbb{F}_{2}^{n}$ of dimension $k$ and is usually defined by a parity-check matrix $\Hm$ of size $r \times n$ as
%$$
%\Cc = \left\{ \xv \in \{0,1\}^n: \Hm \xv^T=\mathbf{0}\right\}.
%$$ 	
%When $\Hm$ is of full rank (which is usually the case) we have $r = n-k$.
%The rate of this code (that we denote by $R$)  is defined as
%$R \eqdef \frac{k}{n}$.
%\newline

%\section{The Quantum Random Oracle Model - QROM}
%\input{qrom.tex} 
\section{The quantum random oracle model}
\label{sec:qqrom}

\subsection{The random oracle model - ROM.}

In many signature schemes we need a function that behaves like a random function. We typically use hash functions to mimic such random functions. The random oracle model (or ROM) is an idealized model that assumes that the hash function used behaves exactly like a random function. This model is appealing as it allows simpler security proofs. There are some specific cases where the ROM is not adapted \cite{CGH04,LN09}. Despite those examples, this model is fairly standard and accepted in the cryptographic community. Particularly, there have been no successful real-world attacks specifically because of the ROM. Additionally, schemes that are proven secure in the ROM are usually efficient.

More precisely, consider a hash function $\mathcal{H} : \zo^n \rightarrow \zo^m$ used in a cryptographic protocol. An adversary would perform an attack by applying $\mathcal{H}$ many times. Suppose the adversary makes $q$ calls to $\mathcal{H}$ on inputs $\xv_1,\dots,\xv_q$ and get answers $\mathcal{H}(\xv_1),\dots,\mathcal{H}(\xv_q)$. In the ROM, this function $\mathcal{H}$ is replaced by a function $f$ uniformly chosen from the set of functions from $\zo^n$ to $\zo^m$. This means that $f$ outputs a random output $\yv_i$ for every input $\xv_i$.

Describing a random function from $\zo^n$ to $\zo^m$ requires $m2^n$ bits and cannot be hence realistically full described. Fortunately, one can emulate queries to a random function $f$ without describing it entirely. We use the following procedure:

\begin{itemize}
	\item[] On input $\xv$, we distinguish $2$ cases: if $\xv$ was queried before then give the same answer, otherwise pick a random $\yv \in \zo^m$ and output $\yv = f(\xv)$.
\end{itemize}
We keep a table of the inputs that were already queried to perform the above procedure, which is efficient. This procedure is especially useful when we want to slightly modify the function $f$, for example by injecting the input of a computational problem as an output of $f$, or more generally to give a special property to $f$.

\subsection{The quantum random oracle model - QROM.}
Since we have hash functions that are believed to be secure against quantum adversaries, it is natural to extend the ROM to the quantum setting. Here again, we assume that we replace the hash function $\mathcal{H} : \zo^n \rightarrow \zo^m$ by a function $f$ uniformly chosen from the set of functions from $\zo^n$ to $\zo^m$. 

What will change compared to the classical setting is the way those functions are queried. Indeed, from the circuit $\mathcal{H}$, it is always possible to construct the unitary $O_\mathcal{H}$ acting on $n+m$ qubits satisfying
$$ \forall \xv \in \zo^n, \forall \yv \in \zo^m, \ O_\mathcal{H}(\ket{\xv}\ket{\yv}) = \ket{\xv}\ket{\mathcal{H}(\xv) + \yv}. $$
When replacing $\mathcal{H}$ with a random function $f$, queries to $O_\mathcal{H}$ are replaced with queries to $O_f$ where
$$ \forall \xv \in \zo^n, \forall \yv \in \zo^m, \ O_f(\ket{\xv}\ket{\yv}) = \ket{\xv}\ket{f(\xv) + \yv}. $$
Again, a random function $f$, and the associated unitary $O_f$ is fully determined by $m2^n$ bits corresponding to all the outcomes $f(\xv)$ for $\xv \in \zo^n$. Unlike the classical case, there is no known procedure to efficiently produce answer to queries. Suppose for example that you want to emulate a query to $O_f$ on input $\frac{1}{2^{n/2}} \sum_{\xv \in \zo^n} \ket{\xv}\ket{\mathbf{0}}$. In order to emulate this, and generate $\frac{1}{2^{n/2}} \sum_{\xv \in \zo^n} \ket{\xv}\ket{f(\xv)}$ we would need to generate some randomness $\rv$ for each $\xv \in \zo^n$. Another way of seeing this difficulty is that the procedure that generates a random number cannot be represented as a circuit and therefore cannot be quantized by the usual procedure.

\subsection{Tweaking the QROM.} As we mentioned, it is often useful to modify the random function and to give it extra properties in order to prove the security of the underlying cryptographic scheme. The fact that we need to emulate $O_f$ makes it much harder to include those changes in an efficient way. There are several known techniques, such as rewinding, reprogramming or challenge injection that can be done in some cases, often with a polynomial loss in the number of challenges. 

Our goal was to limit as much as possible the use of those techniques in order to have the quantum security as tight as possible. The only result we will use is the following from \cite{Zha12}:

\begin{restatable}{proposition}{propqdistrib}
	\label{propo:qdistrib} 
	Say $\cA$ is a quantum algorithm that makes $q$ quantum oracle queries. Suppose further that we draw the oracle $O$ from two distributions. The first is the random  oracle  distribution.  The  second  is  the  distribution  of  oracles  where  the value of the oracle at each input $x$ is identically and independently distributed by some distribution $D$ whose variational distance is within $\eps$ from uniform. Then the variational distance between the distributions of outputs of $\cA$ with each oracle is at most $\frac{8\pi}{\sqrt{3}}q^{\frac{3}{2}}\sqrt{\eps}$.
\end{restatable}
\section{Code-based Full Domain Hash signature schemes}
\label{sec:signatureScheme} 
We give in this section the code-based signatures schemes we will consider in our security proof in the {\QROM} and in \S\ref{subsec:cbProbs} code-based problems that will be involved.

\subsection{Description of the scheme}
\label{subsec:scheme}

Let us first recall the concept of signature schemes.

\begin{definition}
	[Signature Scheme]A signature scheme $\cS$ is a triple of algorithms
	$\Gen$, $\Sgn$, and $\Ver$ which are
	defined as:
	\begin{itemize}
		\item The key generation algorithm $\Gen$ is a probabilistic
		algorithm which given $1^{\lambda}$, where $\lambda$ is the
		security parameter, outputs a pair of matching public and private
		keys $(pk,sk)$;
		\item The signing algorithm is probabilistic and takes as input a
		message $\mv \in \{0,1\}^{*}$ to be signed and returns a signature
		$\sigma = \Sgnsk(\mv)$;
		\item The verification algorithm takes as input a message $\mv$ and
		a signature $\sigma$. It returns $\Vrfypk(\mv,\sigma)$
		which is $1$ if the signature is accepted and $0$ otherwise. It
		is required that \ $\Vrfypk(\mv,\sigma)=1$ if  $\sigma = \Sgnsk(\mv)$.
	\end{itemize}
\end{definition}

We briefly present now the code-based signatures scheme we consider. A binary linear code $\mathcal{C}$ of length $n$ and dimension $k$ (that we denote by $\lbrack n,k \rbrack$-code) is a subspace of $\mathbb{F}_{2}^{n}$ of dimension $k$ and is usually defined by a parity-check matrix $\Hm \in \mathbb{F}_{2}^{(n-k)\times n}$ of full rank as:
$$ 
\mathcal{C} = \{ \xv \in \mathbb{F}_{2}^{n} : \Hm \xv^{T} = \mathbf{0} \}
$$
Code-based signatures schemes we consider are FDH-like in which the following one way function is used:
\begin{displaymath}
\begin{array}{lccc}
f_{\Hm,w} : &  S_{w} & \longrightarrow & \mathbb{F}_{2}^{n-k}\\
& \ev & \longmapsto & \ev\Hm^{T}
\end{array}
\end{displaymath}
where $\Hm \in \mathbb{F}_{2}^{(n-k)\times n}$ is a parity-check matrix of a $\lbrack n,k\rbrack$-code. Inverting this function means on an input $\sv$ (usually called a syndrome) to find an error $\ev$ of Hamming weight $w$ such that $\Hm \ev^{T} = \sv^{T}$. The general scheme is then defined as follows. We first suppose that we have a family of $\lbrack n,k\rbrack$-codes defined by a set of parity-check matrices $\mathcal{F}$ of size $(n-k) \times n$ such that for all $\Hm \in \Fc$ we have an algorithm $\mathcal{D}_{\Hm,w}$ which on input $\sv$ computes $\ev \in f_{\Hm,w}^{-1}(\sv)$. Then we pick uniformly at random $\Hsec \in \Fc$, an $n \times n$ permutation matrix $\Pm$, a non-singular matrix $\Sm \in \mathbb{F}_{2}^{(n-k) \times (n-k)}$ which define the secret and public key as:
$$
sk \leftarrow (\Hsec,\Pm,\Sm) \mbox{ } ; \mbox{ }  pk \leftarrow \Hpub \mbox{ where } \Hpub \eqdef \Sm \Hsec \Pm 
$$

This construction of $\Hpub$ is the standard method to scramble a code and originates from the original work of McEliece \cite{Mce78}.

\begin{remark} Let $\Csec$ be the code defined by $\Hsec$. Then the parity-check matrix $\Hpub$ represents the code
	$\mathcal{C}_{\textup{pub}} \eqdef \{ \cv \Pm : \cv \in \mathcal{C}_{\text{sec}} \}$ 
	with a basis picked uniformly at random thanks to $\Sm$. 
\end{remark}

Then, we select a cryptographic hash function $\mathcal{H} : \{0,1\}^{*} \rightarrow \mathbb{F}_{2}^{n-k}$ and a parameter $\lambda_{0}$ which lead to define algorithms \Sgnsk\ and \Vrfypk\ as follows
\begin{center}
	\begin{tabular}{l@{\hspace{3mm}}|@{\hspace{3mm}}l}
		$\Sgnsk(\mv)\!\!: \qquad \qquad \qquad$ & $\Vrfypk(\mv,(\ev',\rv))\!\!:$ \\
		$\quad \rv \Unif \{ 0,1 \}^{\lambda_{0}}$ &$\quad \sv \leftarrow \hash(\mv |\rv)$ \\
		$\quad \sv \leftarrow \hash(\mv |\rv)$ &$\quad w_{0} \leftarrow |\ev'|$ \\ 
		$\quad \ev \leftarrow \mathcal{D}_{\Hsec,w}(\Sm^{-1}\sv^{T})$ &$\quad \texttt{if } \Hpub \ev'^{T} =\sv^{T} \texttt{ and } w_0 = w \texttt{ return } 1$ \\
		$\quad \texttt{return}(\ev\Pm,\rv)$& $\quad \texttt{else return } 0 $\\			
	\end{tabular} 
\end{center}

To summarize, a signature of a message $\mv$ with the public key $(\Hpub,w)$ is a pair $(\ev,\rv)$ such that $\Hpub \ev^{T} = \mathcal{H}(\mv|\rv)^{T}$ with $|\ev| = w$

\begin{remark}  
	The use of a salt $\rv \in \{0,1\}^{\lambda_{0}}$ in algorithm {\Sgnsk} is made in order to have a tight security proof. In particular, this allows two signatures of a message $\mv$ to be different with high probability.
\end{remark}

\subsection{Code-Based Problems and computational assumtpions}
\label{subsec:cbProbs} 

We introduce in this subsection the code-based problems on which our security reduction in the {\QROM} will stand. The first is Decoding One Out of
Many (DOOM). Its classical version was first considered in \cite{JJ02} and later
analyzed in \cite{S11}. We will come back to its analysis in the quantum case in \S\ref{sec:qdoom}. As we are going to see, the best known algorithms
to solve this problem  are functions of the distance $w$.
Let us first consider the basic problem upon which all code-based cryptography relies.

\problemSD*

This problem has been studied for a long time and despite many efforts on this issue \cite{P62,S88,D91,MMT11,BJMM12,MO15,DT17a} the best known algorithms for solving this problem \cite{BJMM12,MO15} are exponential in the weight $w$ of $\ev$ as long as
$w = (1-\epsilon)(n-k)/2$ for any $\epsilon>0$. 
Furthermore when $w$ is sublinear in $n$, the exponent of the best known algorithms has not changed \cite{CS16} since the Prange algorithm \cite{P62} dating back to the early sixties. Moreover,  it seems very difficult to lower this 
exponent by a multiplicative factor smaller than $\frac{1}{2}$ in the quantum 
computation model as illustrated by 
\cite{KT17}.

However, in a context of code-based signatures an attacker may produce, say $q$, favorable messages and hash them to obtain $s_{1},\cdots,s_{q}$ syndromes on which he tries to solve one of the $q$ instances associated to Problem \ref{prob:SDH}. This brings us to introduce a different version of the SD problem.

\begin{problem}[$\DOOM_q$ -- Decoding One Out of Many]~\\
	\label{prob:DOOM}
	\begin{tabular}{ll}
		Instance: &\quad  $\Hm \in \mathbb{F}_{2}^{(n-k) \times n}$ ; $\sv_{1},\cdots,\sv_{q} \in \mathbb{F}_{2}^{n-k}$	; $w \in \{0,\cdots,n\}$ \\
		Output: &\quad  $(\ev,i) \in \mathbb{F}_{2}^{n}\times \llbracket 1,q \rrbracket$ of Hamming weight $w$ such that $\Hm\ev^{T} = \sv_{i}^{T}$.
	\end{tabular}
\end{problem}

The above problem can be defined for any $q \in \mathbb{N^*}$. This problem is of course easier than $\SD$ but can not be solved at most $q$ times faster than the $\SD$ problem. As it happens best algorithm gain much less than this $q$ factor. Also using the hardness of $\DOOM_q$ is appealing when performing security proofs in the QROM as it allows to avoid instance injection.

Moreover, an interesting feature of the above problem is that known algorithms to solve it fail to take advantage of very large values of $q$. Actually, depending on the parameters, there is a limit after which increasing $q$ does n'ot decrease the time.
Therefore, it is natural to define a variant where we do not limit a priori $q$. We also require the inputs $\sv_i$ to be the output of a random function instead of requiring to write them all. This allows to have a compact description of the inputs. This will also simplify the quantum security proof.

\problemQDOOM*

We study those problems in \S\ref{sec:qdoom}. In the classical setting, we can easily see that those problems are equivalent for sufficiently large values of $q$. We also present there the best known quantum algorithms for \QDOOM.

\begin{definition}[One-Wayness of \QDOOM]
	We define the quantum success of an algorithm $\cA$ against \QDOOM\ with the parameters $n,k,w$ as:
	\begin{align*} 
	QSucc_{\QDOOM}^{n,k,w}\left( \cA \right) = \mathbb{P}\big(  &\cA\left( \Hm,\mathcal{H} \right) \mbox{solution}  \big).
	\end{align*} 
	where $\Hm$ is chosen uniformly at random in $\mathbb{F}_{2}^{(n-k)\times n}$, $\cH$ a hash function in the {\QROM} which takes its values in $\mathbb{F}_{2}^{n-k}$ and the probability is taken over these choices of $\Hm$ and internal coins of $\cA$.

	The quantum computational success in time $t$ of breaking \QDOOM\ with the parameters $n,k,w$ is then defined as:
	\begin{displaymath}
	QSucc_{\QDOOM}^{n,k,w}(t) = \mathop{\max}\limits_{|\cA|\leq t} \left\{
	QSucc_{\QDOOM}^{n,k,w}\left( \cA \right) \right\}
	\end{displaymath}
	%We define in a similar fashion $QSucc^{n,k,w}_{\QDOOM}$.
\end{definition}

As we discussed in the introduction, it is appealing to consider the $\QDOOM$ problem as it will greatly improve our security reduction on the one side, but on the other side remains almost as hard as the SD problem.

\section{Basic security definitions}
\label{sec:preliminaries}

\subsection{Basic definitions}

A function $f(n)$ is said to be negligible if for all polynomials $p(n)$, $|f(n)| < p(n)^{-1}$ for all sufficiently large $n$. We will denote $negl(n)$ the set of negligible functions. 
The statistical distance between two discrete probability distributions over a same space $\mathcal{E}$ is defined as:
$$
\rho(\cD^0,\cD^1) \eqdef \frac{1}{2} \sum_{x \in \mathcal{E}} |\cD^0(x)-\cD^1(x) |.
$$
The following classical proposition on the statistical distance will be useful:
\begin{proposition}
	\label{prop:product}
	Let $(\cD^0_1,\dots,\cD^0_n)$ and $(\cD^1_1,\dots,\cD^1_n)$ be two $n$-tuples of discrete probability distributions where $\cD^0_i$ and $\cD^1_i$ are distributed over a same space $\cE_i$. 
	For $a \in \{0,1\}$, let us denote  by $\cD^a_1 \otimes \dots \otimes \cD^a_n$ the product probability distribution of $\cD^a_1,\dots,\cD^a_n$, that is $\cD^a_1 \otimes \dots \otimes \cD^a_n(x_1,\dots,x_n)=\cD^a_1(x_1) \dots \cD^a_n(x_n)$ with $x_i \in \cE_i$ for $i\in \{1,\dots,n\}$. In such a case we have
	\begin{displaymath}
	\rho\left(\cD^0_1 \otimes \dots \otimes \cD^0_n,\cD^1_1 \otimes \dots \otimes \cD^1_n \right) \leq \sum_{i=1}^n \rho(\cD^0_i,\cD^1_i).
	\end{displaymath}
\end{proposition}

%We are now going to define the concept of distinguisher between two
%distributions and to relate it with the statistical distance.
%
%
%\begin{definition}[Distinguisher]
\noindent A distinguisher between two distributions $\mathcal{D}^{0}$ and
$\mathcal{D}^{1}$ over the same space $\mathcal{E}$ is a randomized
algorithm $\cA$ which takes as input an element of $\mathcal{E}$ that
follows the distribution $\mathcal{D}^{0}$ or $\mathcal{D}^{1}$
outputs $b \in \{0,1\}$. Such an $\cA$ is characterized by its advantage:
\begin{displaymath}
Adv^{\mathcal{D}^{0},\mathcal{D}^{1}}(\cA) \eqdef
\mathbb{P}_{\xi \sim \mathcal{D}^{0}}\left( \cA(\xi) \mbox{
	outputs } 1 \right) - \mathbb{P}_{\xi \sim
	\mathcal{D}^{1}}\left(\cA(\xi) \mbox{ outputs } 1
\right)
\end{displaymath}
where
$\mathbb{P}_{\xi \sim \mathcal{D}^{i}}\left(\cA(\xi) \mbox{
	outputs } 1 \right)$
is the probability that $\cA(\xi)$ outputs $1$ when its inputs
are picked according to the distribution $\mathcal{D}^{i}$ and for
each executions its internal coins are picked uniformly at random. We call
this quantity the advantage of $\cA$ against
$\mathcal{D}^{0}$ and $\mathcal{D}^{1}$.
%\end{definition}
%We are now able to define the computational distance between two
%distributions.

\begin{definition}
	[Quantum Computational Distance and Indistinguishability]
	The quantum computational distance between two distributions $\mathcal{D}^{0}$ and $\mathcal{D}^{1}$ in time $t$ is: 
	\begin{displaymath}
	\rho_{Qc}\left( \mathcal{D}^{0},\mathcal{D}^{1}\right)(t) \eqdef
	\mathop{\max}\limits_{ |\cA| \leq t} \left\{
	Adv^{\mathcal{D}^{0},\mathcal{D}^{1}}(\cA) \right\}
	\end{displaymath}
	where $|\cA|$ denotes the running time of $\cA$ on
	its inputs.
	
	The ensembles $\mathcal{D}^{0}=(\mathcal{D}^{0}_{n})$ and
	$\mathcal{D}^{1} = (\mathcal{D}_{n}^{1})$ are computationally
	indistinguishable in time $(t_{n})$ if their computational distance
	in time $(t_{n})$ is negligible in $n$.
\end{definition}
%In other words, the quantum computational distance is the best advantage that any quantum adversary could get in bounded time. It is well known that statistical distance is greater than
%quantum computational distance as the following theorem claims.
%
%\begin{proposition}\label{pr:comp_stat}
%	Let $\mathcal{D}^{0}$ and $\mathcal{D}^{1}$ be two distributions, then $\rho\left( \mathcal{D}^{0},\mathcal{D}^{1} \right)$ is the best advantage that any adversary could get, even with unbounded time:
%	\begin{displaymath}
%	\forall t, \quad \rho_{Qc}\left( \mathcal{D}^{0},\mathcal{D}^{1}
%	\right)(t) \leq \rho\left( \mathcal{D}^{0}, \mathcal{D}^{1}
%	\right).
%	\end{displaymath}
%\end{proposition}
\subsection{Digital signature security and games.}

For signature schemes one of the strongest security notion is {\em
	Quantum Existential Unforgeability under an adaptive Chosen Message Attack}
(QEUF-CMA). In other words, a quantum adversary has access to any signatures
of its choice and its goal is to produce a valid forgery. A valid
forgery is a message/signature pair $(\mv,\sigma)$ such that
$\Vrfypk(\mv,\sigma)=1$ whereas the signature of $\mv$ has never been
requested by the forger. Moreover the forger has access to quantum hash queries. By quantum hash queries we mean that adversaries can make a superposition of queries. In other words, a quantum access to a hash function $\mathcal{H}$ is an access to the following oracle:
\[ O_{\mathcal{H}} : \ket{\mv,\zv} \mapsto\ket{\mv,\zv \oplus \mathcal{H}(\mv)} \]
Let us now define the QEUF-CMA security of a signature scheme:

\begin{definition}
	[\textup{QEUF-CMA} Security] Let $\cS$ be a signature scheme.\\  A forger $\cA$
	is a $(t,\qhash,\qsig,\varepsilon)$-adversary in \textup{QEUF-CMA} against
	$\cS$ if after at most $\qhash$ quantum-queries to the hash oracle, $\qsig$
	classical-queries to signing oracle and $t$ working time, it outputs a valid forgery
	with probability at least $\varepsilon$.
	We define the \textup{QEUF-CMA} success probability against $\cS$ as:
	\begin{displaymath}
	QSucc_{\cS }^{\textup{QEUF-CMA}}(t,\qhash,\qsig) \eqdef
	\max \left( \varepsilon \mbox{} | \mbox{it exists a }
	(t,\qhash,\qsig,\varepsilon) \mbox{-adversary} \right).
	\end{displaymath}
	The signature scheme $\cS$ is said to be
	$(t,\qhash,\qsig)$-secure in \textup{QEUF-CMA} if the above success
	probability is a negligible function of the security parameter
	$\lambda$.
	
\end{definition}

In order to prove that a signature scheme is \textup{QEUF-CMA} under some assumptions we will use the paradigm of games. A good reference of this topic can be found in \cite{S04a}. The following game gives the \textup{QEUF-CMA} security:

\begin{definition}[challenger procedures in the \textup{QEUF-CMA} Game]
	Challenger procedures for the \textup{QEUF-CMA} Game corresponding to a signature scheme
	$\cS$ are defined as:
	\small
	\begin{center} \tt
		\begin{tabular}{|l|l|l|l|}
			\hline
			\underline{proc Initialize$(\lambda)$} & \underline{proc Hash$(\mv,\rv)$} & \underline{proc Sign$(\mv)$} & \underline{proc Finalize$(\mv,\sigma)$} \\ 
			$(pk,sk) \leftarrow \Gen(1^{\lambda})$ & {return} $\hash(\mv)$ & return $\Sgnsk(\mv)$ & return ($\Vrfypk(\mv,\sigma) = 1$) \\
			return $pk$ & & & \\
			\hline  
		\end{tabular}
	\end{center}
\end{definition}

\section{Quantum security of FDH-like code-base signature schemes}
\label{sec:securityProof} 
In this section, we show that code-based signature schemes we defined in \S\ref{sec:signatureScheme} are QEUF-CMA in the $\QROM$ against quantum adversaries. We redescribe the most important aspects of the scheme $S_{code}$ defined in \S\ref{sec:signatureScheme} so that the proof is easier to follow.

We have a family of $\lbrack n,k\rbrack$-codes defined by a set of parity-check matrices $\mathcal{F}$ of size $(n-k) \times n$ such that for all $\Hm \in \Fc$ we have an algorithm $\mathcal{D}_{\Hm,w}$ which on input $\sv$ computes $\ev \in f_{\Hm,w}^{-1}(\sv)$ where $f_{\Hm,w}$ is the function such that $f_{\Hm,w}(\ev) = \ev\Hm^{T}$. Then we pick uniformly at random $\Hsec \in \Fc$, an $n \times n$ permutation matrix $\Pm$, a non-singular matrix $\Sm \in \mathbb{F}_{2}^{(n-k) \times (n-k)}$. The secret and public key are:
$$
sk \leftarrow (\Hsec,\Pm,\Sm) \mbox{ } ; \mbox{ }  pk \leftarrow \Hpub \mbox{ where } \Hpub \eqdef \Sm \Hsec \Pm 
$$

\noindent The signing and verification procedures are then the following 

\begin{center}
	\begin{tabular}{l@{\hspace{3mm}}|@{\hspace{3mm}}l}
		$\Sgnsk(\mv)\!\!: \qquad \qquad \qquad$ & $\Vrfypk(\mv,(\ev',\rv))\!\!:$ \\
		$\quad \rv \Unif \{ 0,1 \}^{\lambda_{0}}$ &$\quad \sv \leftarrow \hash(\mv |\rv)$ \\
		$\quad \sv \leftarrow \hash(\mv |\rv)$ &$\quad w_{0} \leftarrow |\ev'|$ \\ 
		$\quad \ev \leftarrow \mathcal{D}_{\Hsec,w}(\Sm^{-1}\sv^{T})$ &$\quad \texttt{if } \Hpub \ev'^{T} =\sv^{T} \texttt{ and } w_0 = w \texttt{ return } 1$ \\
		$\quad \texttt{return}(\ev\Pm,\rv)$& $\quad \texttt{else return } 0 $\\			
	\end{tabular} 
\end{center}

\noindent Let us first recall and give definitions of distributions that will be used:
\begin{itemize}
	\item $\cU_{w}$ is the uniform distribution over $S_w$ (words of weight $w$).

	\item $\cU_{n-k}$ is the uniform distribution over $\F_2^{n-k}$.

	\item $\cD_w$ is the distribution of $\cD_{\Hsec,w}(\Sm^{-1}\sv^{T})$ when $\sv \Unif \mathbb{F}_{2}^{n-k}$ where $\cD_{\Hsec,w}(\cdot)$ is the algorithm used in $\cS_{\textup{code}}$ to invert $\ev \in S_{w}\mapsto \ev \Hm^{T}_{\textup{sec}}$.
	
	\item $\cD_{w}^{\Hpub}$ is the distribution of the syndrome $\Hpub \ev^{T}$ where $\ev$ is drawn uniformly at random in $S_{w}$

	\item $\Dpub$ is the distribution of public keys $\Hpub$.

	\item $\Drand$ is the uniform distribution over parity-check matrices of size $(n-k)\times n$. 
\end{itemize}

\noindent Our main security statement is the following

\begin{theorem}[Security Reduction]
	\label{theo:secRedu}
	Let $\mathcal{S}_{code}$ be the signature scheme defined in \S\ref{sec:signatureScheme} with security parameter $\lambda$. 
	Let $\qhash$ (resp. $\qsig$) be the number of queries to the hash
	(resp. signing) oracle. We also take $\lambda_0 = \lambda + 2\log_{2}(\qsig)$. For any running time $t$ we have 
	\begin{multline*}
	QSucc_{\cS_{\text{code}}}^{\textup{QEUF-CMA}}(t,\qhash,\qsig) \leq 
	2 \cdot QSucc_{\QDOOM}^{n,k,w}(2t) + \\ \rho_{Qc} \left( \Dpub ,\Drand \right)(2t) + \frac{8\pi}{\sqrt{3}} \qhash^{\frac{3}{2}} \sqrt{\mathbb{E}_{\Hpub} \left( \rho(\cD_{w}^{\Hpub},\mathcal{U}_{n-k}) \right)} + \qsig \rho\left( \cU_{w},\cD_{w}  \right)   + \frac{1}{2^\lambda}
	\end{multline*}
\end{theorem}

In other words, signature schemes we introduced in \S\ref{sec:signatureScheme} can be reduced to the hardness of $\QDOOM$ in the {\QROM} if the family $\mathcal{F}$ and the signature scheme satisfy the following conditions:

\begin{restatable}{condition}{cdt}$ $
	\label{cdt} 
	
	\begin{enumerate} 
		
		\item $\frac{8\pi}{\sqrt{3}} \qhash^{\frac{3}{2}} \sqrt{\mathbb{E}_{\Hpub} \left( \rho(\cD_{w}^{\Hpub},\mathcal{U}_{n-k}) \right)}) \in negl(\lambda)$
		\item $\qsig \rho\left( \cU_{w},\cD_{w}  \right) \in negl(\lambda)$
		\item $\rho_{Qc} \left( \Dpub ,\Drand \right)(t) = o(\frac{t}{2^{\lambda}})$. 
	\end{enumerate}
\end{restatable}

The two first properties are properties of the code family $\mathcal{F}$ used while the third property is a property on the signing algorithms used: we require that signatures which are produced are indistinguishable from words uniformly and independently picked in $S_{w}$.

Notice that our security reduction is almost tight if the above holds. Indeed, we double the running and lose a factor $2$ in front of $QSucc_{\QDOOM}^{\textup{QEUF-CMA}}(t,\qhash,\qsig)$. This makes us lose $2$ bits of security. Actually, we could have a really tight reduction but it would involve a huge amount of quantum memory and access to quantum RAM. We wanted to construct an algorithm in our reduction in the most efficient way so we avoided this solution. We discuss this more at the end of the section. \\

\COMMENT{
	under a set of parameters which achieve the following points:
	\begin{enumerate}
		\item The problem $\QDOOM$ is hard against quantum adversaries.
		\item The distributions $\Drand$ and $\Dpub$ are quantum computationally indistinguishable.
		\item Signatures which are produced are indistinguishable of words uniformly and independently picked in $S_{w}$. 
		\item Syndromes $\Hpub \ev^{T}$ where $\ev$ is drawn uniformly at random in $S_{w}$ are in average on matrices $\Hpub$ indistinguishable of random words in $\mathbb{F}_{2}^{n-k}$. 
	\end{enumerate}
}

The goal of what follows is to prove Theorem \ref{theo:secRedu}. Our security reduction will go as follows: let $\cA$ be a $(t,\qsig,\qhash,\varepsilon)$-quantum adversary in the QEUF-CMA model against $\cS_{\textup{code}}$. 
Recall that in the QEUF-CMA model, we have a benign challenger and the following procedures

\begin{center} \small  \tt
	\begin{tabular}{|l|l|l|l|}
		\hline
		\underline{proc Initialize$(\lambda)$} & \underline{proc Hash$(\mv,\rv)$} & \underline{proc Sign$(\mv)$} & \underline{proc Finalize$(\mv,\ev,\rv)$} \\ 
		$(pk,sk,\lambda_{0}) \leftarrow \Gen(1^{\lambda})$ & {return} $\hash(\mv|\rv)$ & $\rv \Unif \{0,1\}^{\lambda_{0}}$ & $\sv \leftarrow {\tt Hash}(\mv,\rv)$ \\
		$sk \leftarrow (\Pm,\Sm,\Hsec)$ & & $\sv \leftarrow$ \texttt{Hash}$(\mv,\rv)$ & return \\
		$pk \leftarrow (\Hpub \eqdef \Sm \Hsec \Pm$& & $\ev \leftarrow \mathcal{D}_{\Hsec,w}(\Sm^{-1}\sv^{T})$ & $\Hpub \ev^{T} = \sv^{T} \wedge \wt{\ev} = w$  \\
		return $(\Hpub,w)$ & & return $(\ev\Pm,\rv)$ &  \\
		\hline  
	\end{tabular}
\end{center}

In this model, $\cA$ performs the following actions, that we model by a game:
\horizontall
\begin{center}
	\vspace*{-0.3cm}\textbf{Game 0} 
	\begin{enumerate}
		\item $\cA$ makes a call to $\Initialize(\lambda)$ and receives $\Hpub$.
		\item $\cA$ performs $\qsig$ calls to the $\Sign$ procedure. Let $\mv_i$ the message that $\cA$ wants to sign at query $i$ and let $\sigma_i$ the corresponding signature answered by the challenger.
		\item $\cA$ performs an algorithm that makes $\qhash$ calls to $\Hash$ and outputs $\mv',\ev',\rv'$
		\item $\cA$ wins if $\forall i, \ \mv_i \neq \mv'$ and $\Finalize(\mv',\ev',\rv') = 1$. This happens with probability $\eps$ and the whole running time is $t$.
	\end{enumerate}
\end{center}
\horizontall

Recall that procedure $\Sign$ is done by the challenger and $\cA$ queries the challenger. $\cA$ does not have access to the secret key and cannot run $\Sign$ by himself. Procedure $\Hash$ is public, efficient and is used both by the challenger and the adversary $\cA$.

Our security reduction will go as follows: from the adversary $\cA$, we will construct an algorithm $\cB$ to solve the $\QDOOM$ problem. The main part of the proof will be to replace the hash function $\cH$ (modeled by a random function from the QROM) by another hash function that we call $Z$. In Subsection \ref{subsec:Z} we show how to construct this function and in Subsection \ref{sec:securityProof3} we prove our main security statement.

\subsection{Constructing the hash function $Z$}
\label{subsec:Z}
Informally, we want the following properties for $Z$:
\begin{enumerate}
	\item $Z$ is statistically close to a random function in the QROM.
	\item $Z$ and $O_Z$ can be computed efficiently
	\item For any message $\mv$, there is an efficient algorithm to construct $\rv \in \F_2^{\lambda_0}$ and $\ev \in S_w$ such that $Z(\mv,\rv) = \Hpub\ev^{T}$ \emph{without knowing the secret key $\Sm,\Pm,\Hsec$}.
	\item With constant probability, $Z(\mv,\rv) = \mathcal{H}(\mv,\rv)$.	
\end{enumerate}

\noindent The first two properties will allow us to replace calls to $O_\cH$ with calls to $O_Z$ in $\cA$ without changing much the statistical distance of the output. The third property will then allow to change the signing oracle into one that can be done locally without knowing the secret key. The final property will still enforce that the algorithm $\cB$ we construct indeed solves the $\QDOOM$ problem.

\subsubsection{Construction of $Z$.}
Let $J$ be a cryptographic hash function that takes its values in $\F_2 \times S_w$. In particular, the first bit of $J(\mv,\rv)$ is a random element of $\mathbb{F}_{2}$. From the functions $J$ and $\cH$ we can build the function $Z : \F_2^{*} \rightarrow \mathbb{F}_2^{n}$ as follows: fix an input $(\mv,\rv)$ and let $(b,\ev) = J(\mv,\rv)$. If $b = 0$ then $Z(\mv,\rv) = \cH(\mv,\rv)$ else $Z(\mv,\rv) = \Hpub\ev^{T}$. We can easily construct an efficient quantum circuit for $O_Z$ using $O_\cH$ and $O_J$. For the running time of $O_Z$, we assume that the running time of $\mathcal{H}$ is roughly equivalent to the computing time of $(\Hpub\ev^{T})$ (if this is not the case, we can use a slower hash function $\mathcal{H}$ to match those $2$ times).

\begin{proposition}
	\label{Proposition:Z_close_to_H}
	For any $\Hpub$, outputs of $Z$ are at most at statistical distance $\rho(\cD_{w}^{\Hpub},\mathcal{U}_{n-k})$ to outputs of a random function in the $QROM$. 
	%where the probability is taken over matrices $\Hpub$.
\end{proposition}
\begin{proof} It directly follows from the definition of $Z$ and $\cD_{w}^{\Hpub}$ given above. Indeed, for any input $(\mv,\rv)$, if $J(\mv,\rv) = (0,\ev)$ then the output distribution is totally random and equal to $\mathcal{U}_{n-k}$. Otherwise, it follows the distribution of $\cD_{w}^{\Hpub}$. Each of these events happens with probability $\frac{1}{2}$ which concludes the proof.
\end{proof}

Moreover, for any message $\mv$, we can find $\rv \in \F_2^{\lambda_0}$ and $\ev \in S_w$ such that $Z(\mv,\rv) = \Hpub\ev^{T}$ with the following procedure: find $\rv_0$ such that $J(\mv,\rv_0) = (b,\ev_0)$ with $b=1$. This means that the running time of $O_Z$ is twice the running time of $O_{\mathcal{H}}$. This can be done with $2$ calls to $J$ on average. Output $\rv_0,\ev_0$ and notice that $\ev_0 = Z(\mv,\rv_0)$.

\subsection{Proof of Theorem \ref{theo:secRedu}} 
\label{sec:securityProof3}

\begin{proof}
	We present a sequence of games which initiates with Game $0$ presented at the beginning of this section and ends with an quantum algorithm solving the $\QDOOM$ problem. Let $(\Hm_{0},\mathcal{H})$ be an instance of the $\QDOOM$-problem for parameters $n,k,w$ given by $\cS_{\text{code}}$. We will denote by $\mathbb{P}\left( S_{i} \right)$ the probability of success of the game $i$. 
	\bigskip

	{\bf Game $1$} is identical to Game $0$ except that we change the winning condition. Let $F$ be the following failing event: there is a collision in a signature query ({\em i.e.} two signatures queries for a same message $\mv$ lead to the same salt $\rv$). The adversary wins Game $1$ only if $F$ does not occur additionally to the other requirements. A direct application of the birthday paradox gives $\mathbb{P} \left( F \right) \le \frac{1}{2^\lambda}$ and 
	\begin{displaymath}
	\mathbb{P}\left( S_{0} \right) \leq \mathbb{P}\left( S_{1} \right) - \mathbb{P} \left( F \right) \leq \mathbb{P}\left( S_{1} \right) - \frac{1}{2^\lambda}. 
	\end{displaymath}

	{\bf Game $2$.} 
	Here, we consider Game $1$ but both the adversary and the challenger use a different procedure {\tt
		Hash}. \COMMENT{We replace $\mathcal{H}$ by $Z$. From our previous discussion, this will allow to later emulate the signature queries. Moreover, since with probability $\frac{1}{2}$, it has the same output as $\mathcal{H}$,  The idea may be quickly described as follows: we are going to replace $\mathcal{H}$ by a function which outputs instances of $\QDOOM$ or syndromes that we know how to decode depending on a process that we control. In this way we will force the adversary, with a high probability, to solve $\QDOOM$ while ensuring to be able to answer to its signatures queries. 
		%	
		%	To achieve this goal we recall that $G$ follows the QROM that enables to a quantum access to the hash function $G$ meaning that we have access to the following oracle
		%	$$ O_G : \ket{\mv,\rv,\zv} \rightarrow \ket{\mv,\rv,\zv \oplus G(\mv,\rv)}. $$
		%			where $G$ is a function from $\F_2^{m + \lambda_0}$ to $\F_2^n$.
		For this purpose we are going to use the function $Z$, for which we have an efficient quantum oracle access $O_Z$, defined in \S\ref{subsec:Z}. Recall that $Z(\mv,\rv) = \Hpub\ev^{T}$ if $J(\mv,\rv) = (1,\ev)$ for some $\ev$ and $Z(\mv,\rv) = \mathcal{H}(\mv,\rv)^{T}$ otherwise.  }
	The $\Hash(\mv,\rv)$ procedure hence becomes: return $Z(\mv,\rv)$. A call to $O_{\Hash}(\ket{\psi})$ returns similarly $O_Z(\ket{\psi})$ for all $\ket{\psi}$. We can relate this game to the previous one through the following lemma.
	\begin{restatable}{lemma}{lemdistribi}
		\label{lem:distribi}
		\begin{displaymath}
		\mathbb{P}(S_{1})\leq \mathbb{P}(S_{2}) + \frac{8\pi}{\sqrt{3}} \qhash^{\frac{3}{2}}\sqrt{\mathbb{E}_{\Hpub} \left( \rho(\cD_{w}^{\Hpub},\mathcal{U}_{n-k}) \right)} 
		\end{displaymath}
	\end{restatable}
	%	\begin{proof}
	%		\anote{convexity?}
	%		We know from Proposition \ref{Proposition:Z_close_to_H} that in the $QROM$, $Z$ is $\mathbb{E}\left(\rho(\Hpub(\mathcal{U}_{w}),\cU)\right)$ close to a random function. Game $2$ differs from game $1$ by replacing each call to $\Hash$ (resp. $O_{\Hash}$) by a call to $Z$ (resp. $O_Z$). Using Proposition \ref{propo:qdistrib}, the output state after game $2$ differs (in statistical distance) from the output state after game $1$ by at most $4q^2\sqrt{\mathbb{E}\left(\rho(\Hpub(\mathcal{U}_{w}),\cU)\right)}$. From there, we can conclude that 
	%$
	%		\mathbb{P}(S_{1})\leq \mathbb{P}(S_{2}) + 4q^{2} \sqrt{\eps_0}.
	%$
	%	\end{proof}	
	\begin{proof} It is clear that $\mathbb{P}(S_{1}) - \mathbb{P}(S_{2}) = \mathbb{E}_{\Hpub}\left( \mathbb{P}\left( S_{1} | \Hpub \right) -  \mathbb{P}\left( S_{2} | \Hpub \right) \right)$. Moreover if we fix $\Hpub$, we know from Proposition \ref{Proposition:Z_close_to_H} that in the $QROM$, outputs of $Z$ are at most at distance $\rho(\cD_{w}^{\Hpub},\cU_{n-k})$ from uniform. Game $2$ differs from game $1$ by replacing each call to $\Hash$ (resp. $O_{\Hash}$) by a call to $Z$ (resp. $O_Z$). Using Proposition \ref{propo:qdistrib}, the output state after game $2$ differs (in statistical distance) from the output state after game $1$ by at most $\frac{8\pi}{\sqrt{3}} \qhash^{\frac{3}{2}} \sqrt{ \rho(\cD_{w}^{\Hpub},\cU_{w})}$ which leads to:
		$$
		\mathbb{P}\left( S_{1} | \Hpub \right) -  \mathbb{P}\left( S_{2} | \Hpub \right) \leq \frac{8\pi}{\sqrt{3}} \qhash^{\frac{3}{2}} \sqrt{ \rho(\cD_{w}^{\Hpub},\cU_{w})}
		$$ 
		Then by concavity of the root function and Jensen's inequality we get:
		$$
		\mathbb{P}(S_{1}) - \mathbb{P}(S_{2}) \leq \frac{8\pi}{\sqrt{3}} \qhash^{\frac{3}{2}} \sqrt{ \mathbb{E}_{\Hpub} \left( \rho(\cD_{w}^{\Hpub},\cU_{w}) \right)}
		$$

	\end{proof}

	%	lemma given in (see\cite[Lemma 3]{BDF+11}.
	%	
	%	
	%	\begin{restatable}{lemma}{lemqdistrib}
	%		\label{lem:qdistrib} 
	%		Say $\mathcal{A}$ is a quantum algorithm that makes $q$ quantum oracle queries. Suppose further that we draw the oracle $O$ from two distributions. The first is the random oracle distribution. The second is the distribution of oracles where the value of the oracle at each input x is identically and independently distributed by some distribution $\mathcal{D}$ whose variational distance is within $\varepsilon$ from uniform. Then the variational distance between the distributions of outputs of $\mathcal{A}$ with each oracle is at most $4q^{2} \varepsilon$. 
	%	\end{restatable} 
	%
	%	We show in appendix how to emulate the lists $\listM$ in such a way
	%	that list operations cost, including its construction, is at most
	%	linear in the security parameter $\lambda$. Since $\lambda\le n$, it
	%	follows that the cost to a call to {\tt proc Hash} cannot exceed
	%	$O(n^2)$ and the running time of the challenger is
	%	$t_{c} = t +  O\left( \qhash \cdot n^{2} \right)$.	
	{\bf Game $3$} differs from Game $2$ by changing in {\tt proc
		Sign}. When it is queried $\mv$, the procedure ``$\ev\gets \mathcal{D}_{\Hsec,w}(\Sm^{-1}\sv^{T})$, return $(\ev\Pm,\rv)$'' is replaced by
	``find $(\ev,\rv)$ such that $J(\mv,\rv) = (1,\ev)$, return $(\ev,\rv)$''.

	Any signature $(\ev,\rv)$ produced by
	{\tt proc Sign} is valid. $J$ is modeled as a random function so the error $\ev$ is drawn according to
	the uniform distribution $\mathcal{U}_{w}$ while previously it was
	drawn according to the output distribution of $\cD_{\Hsec,w}$. We therefore have thanks to Proposition \ref{prop:product}
	\begin{displaymath}
	\mathbb{P} \left( S_{2}
	\right) - \mathbb{P} \left( S_{3}
	\right) \leq \qsig \rho(\cU_w,\cD_w)
	\end{displaymath}
	Moreover, to find $\rv$ such that $J(\mv,\rv) = (1,\cdot)$ we pick uniformly at random $\rv$ until finding it. As outputs of $J$ are uniformly distributed, we find such a $\rv$ in a constant time.\smallskip
	\bigskip
	
	{\bf Game $4$} is the game where in the initialize procedure, we replace the public matrix
	$\Hpub$ by $\Hm_{0}$, which is a totally random matrix in $\F_2^{(n-k) \times n}$.
	In this way we will force the adversary to
	build a solution of the \QDOOM\ problem. Here
	if a difference is detected between games it gives a
	distinguisher between the distribution $\Drand$
	and $\Dpub$:
	\begin{displaymath}
	\mathbb{P} \left( S_{3} \right) \leq \mathbb{P} \left( S_{4} \right) + \rho_{Qc} \left( \Dpub,\Drand \right)\left(2t \right).   
	\end{displaymath}
	
	{\bf Game $5$} differs in the finalize procedure as follows:
	\begin{center} 
		\begin{tabular}{|l|}
			\hline
			\underline{{\tt proc Finalize}}$(\mv ,\ev,\rv)$ \\
			$\sv\gets {\tt Hash}(\mv,\rv)$ \\
			$b \leftarrow \Hpub\ev^{T} = \sv^{T} = 0 \wedge |\ev| = w$ \\  
			$(b',\ev) = J(\mv,\rv)$ \\ 
			return $b \wedge (b'==0)$ \\
			\hline
		\end{tabular}
	\end{center}
	
	We assume the forger outputs a valid signature $(\ev,\rv)$ for
	the message $\mv $. The probability of success of Game $5$ is the
	probability of the event ``$S_4 \wedge(J(\mv,\rv) = (0,\ev))$''.
	
	If the forgery is valid, the message $\mv $ has never been queried by
	{\tt Sign}, and the adversary never had access to any
	output of $J(\mv,\cdot)$. This way, the two events are
	independent and we get:
	$$
	\mathbb{P} \left( S_{5} \right) = \mathbb{P}_{\mv,\rv} \left( J(\mv,\rv) = (0,\ev) \right) \cdot  \mathbb{P} \left( S_{4} \right) = \frac{1}{2} \mathbb{P} \left( S_{4} \right).
	$$
	The probability $\mathbb{P} \left( S_{5} \right)$ is then exactly the probability for $\cA$ to output $\mv,\rv$ and $\ev \in S_{w}$ such that $\Hm_{0}\ev^{T} = \mathcal{H}(\mv,\rv)^{T}$ which gives
	\begin{align*}\label{eq:upper_bound_p5}
	\mathbb{P} \left( S_{5} \right) \leq QSucc_{\QDOOM}^{n,k,w}(2t).
	\end{align*} 
	as we know thanks to the output a preimage $(\mv,\rv)$ of the solution of the decoding problem. 
	This  concludes the proof of Theorem \ref{theo:secRedu} by combining this together with all the bounds obtained for each of the previous games.
\end{proof} 

\subsection*{Why do use the random function $Z$ to reprogram our random oracle?}
We just want to briefly mention why we use an extra function $J$ to reprogram our (quantum) random oracle. We could have just, for the $q$ values we use, reprogram the function $\mathcal{H}$ accordingly, as it is done for example in \cite{ABB+17}. However, this actually requires $q$ extra quantum bits of memory (recall that $q = 2^\lambda$ and can be very large) as well as an efficient quantum data structure that would act as a quantum RAM. However, we do not have yet efficient models of quantum RAM, as shown in \cite{AGJ+15}. We do not want to go to deep in the discussion whether such data structures in the quantum model should be allowed or not, this is work for future research. However, we want to be in the safe side of things by not allowing here this kind of data structures. This means in particular that our reduction from the adversary $\mathcal{A}$ that breaks the signature scheme to the algorithm $\mathcal{B}$ that solves the {\QDOOM} problem not only preserves essentially the quantum time but also more generally the quantum resources used, in particular quantum memory.
\section{{\QDOOM} Study}
\label{sec:qdoom} 
	We study here the best known quantum algorithms to solve $\QDOOM$. They all come from an old algorithm due to Prange \cite{P62} and are known as Information Set Decoding (ISD). These kind of algorithms were first thought to solve the SD problem. The current state-of-the-art to solve the ${\DOOM_q}$ and {$\DOOM_{\infty}$} slightly adapt them. In this way we are first going to describe general a skeleton of ISDs and quantum algorithms in this setting. Moreover, during our discussion we  will give several reasons on why we think it is difficult to improve significantly quantum algorithms using ISDs.

\subsubsection{Notations} We provide here some notations that will be used throughout this section. Let $\Hm$ be a matrix of size $(n-k) \times n$ in $\mathbb{F}_{2}$ and $I = \{i_{1},\cdots,i_{p}\} \subseteq \{1,\cdots,n\}$. We define the permutation $\pi_{I}$ as:
$$
\pi_{I}(i_{j}) = j \mbox{ for } 1 \leq j \leq p \mbox{ and } \pi_{I}(j) = j \mbox{ otherwise}
$$ 
and $\Pm_{\pi_{I}}$ its associated matrix. Then $\Hm_{\pi_{I}}$ will denote $\Hm \Pm_{\pi_{I}}$. All quantities we are interested in are functions of the code-length $n$ and we will write $f(n) = \tilde{O}(g(n))$ when there exists a constant $C$ such that $f(n) = O\left( \log_{2}^{C}(g(n))\cdot g(n) \right)$ and $f(n) = \Theta\left(g(n)\right)$ when there exists two constants $m,M$ such that $mg(n) \leq f(n) \leq Mg(n)$.

\subsection{Information Set Decoding - ISD}

Let us first recall that algorithms we will study were thought to solve the following problem:
\problemSD*

\noindent Existing literature in the study of algorithms solving $\SD$ usually assumes that there is a unique solution as for instance in a context of encryption the ciphertext of $\ev$ is $\Hm \ev^{T}$ (see \cite{N86}) which imposes to have an injective construction. In the case of code-based signature schemes we introduced in \S\ref{subsec:scheme}, the weight $w$ is chosen greater than the Gilbert-Varshamov bound, namely 
$d_{\textup{GV}}(n,k) \eqdef n h^{-1}(1-k/n)$ where $h(x) \eqdef -x \log_2(x) -(1-x) \log_2(1-x)$ and $h^{-1}(x)$ is the inverse function defined for $x$ in 
$[0,\frac{1}{2}]$ and ranging over $[0,1]$. It represents the weight $w$ for which we can typically expect that SD admits one solution, beyond it there typically exits an exponential number of solutions and below it no solution. We need to choose $w$ greater than this bound in order to be able to invert the function $\ev \in S_{w} \mapsto \ev\Hm^{T}$ on all words of $\mathbb{F}_{2}^{n-k}$. More precisely, the following proposition gives the number of solutions which are expected:

\begin{proposition} 
	\label{propo:numbS} 
	Let $w$ be an integer and $\sv \in \mathbb{F}_{2}^{n-k}$, then there exists in average $M_{n,k,w} \eqdef \frac{\binom{n}{w}}{2^{n-k}}$ solutions to $\SD$ where probabilities are taken by picking matrices $\Hm$ uniformly at random in $\mathbb{F}_{2}^{(n-k)\times n}$. 
\end{proposition} 

\begin{remark} Asymptotically $\binom{n}{w} = \tilde{O} \left( 2^{n \cdot h(w/n)} \right)$, then Gilbert-Varshamov's bound easily gives the weight for which we expect in average one solution to $\SD$. 
\end{remark}

In the following we will consider weights $w$ greater than $d_{\textup{GV}}(n,k)$ and we will have to take into account $M_{n,k,w}$ in our study.

\subsubsection{The Prange Algorithm.} Let us first consider a $\lbrack n,k \rbrack$-code $\mathcal{C}$ with parity-check matrix $\Hm \in \mathbb{F}_{2}^{(n-k)\times n}$ and a syndrome $\sv \in \mathbb{F}_{2}^{n-k}$. 
The matrix $\Hm$ is a full-rank, therefore we can choose uniformly at random a set $I \subseteq \{1,\cdots,n\}$ of size $n-k$, usually called an information set, such that, with a high probability, $\Hm$ restricted to these positions is an invertible matrix. In other words we have $\Hm_{\pi_{I}} = \lbrack \Am|\Bm\rbrack$ where $\Am \in \mathbb{F}_{2}^{(n-k)\times (n-k)}$ is non-singular. We look now for $\ev$ of the form $\ev_{\pi_{I}} = (\ev'|\mathbf{0}_{k})$. We should therefore have $\sv^{T} = \Hm \ev^{T} = \Am \ev'^{T}$. Then thanks to Gaussian elimination, which is done in polynomial time, we compute $\ev'^{T}=\Am^{-1} \sv^{T}$. In this way, if the weight of $\ev = (\ev',\mathbf{0}_{k})_{\pi_{I}^{-1}}$ is $w$, we just found a solution, otherwise we pick an other set $I$ of $n-k$ positions. 
Thus, the hard part of this algorithm consists of finding the good set of positions. It can be shown that the probability to find a fixed error of weight $w$ during an iteration is given by $p_{prange} \eqdef \frac{\binom{n-k}{w}}{\binom{n}{w}}$ (it relies among other things on a counting argument over information sets). As it is explained above there is an exponential number $M_{k,n,w}$  (see Proposition \ref{propo:numbS}) errors $\ev$ of weight $w$ such that $\Hm \ev^{T} = \sv^{T}$. In this way, under the assumption (which is a classical one in the study of ISDs) that solutions to SD behave independently of the set $I$ we pick,
the average probability (on matrices $\Hm$) to not find any solution during an iteration is $(1-p_{prange})^{M_{k,n,w}}$ which implies a probability of succeed during one iteration:
$$
P_{prange} \eqdef 1 - (1-p_{prange})^{M_{k,n,w}} = \Theta \big( \min \left( 1, M_{n,k,w} \cdot p_{prange} \right) \big) 
$$
$\mbox{ where } M_{n,k,w} \cdot p_{prange} = \frac{\binom{n-k}{w}}{2^{n-k}}$. Thus, Prange's algorithm will make on average $\tilde{O}\left( 1
/P_{prange} \right)$ samples which gives its complexity as the Gaussian elimination is polynomial and it is easily verified that for all $w$ such that $d_{\textup{GV}}(n,k) \leq w < (n-k)/2$, $1/P_{prange}$ is exponential in the code length.

\subsubsection{Quantum quadratic speedup of the Prange algorithm.} There is a direct quantum quadratic speedup which consists to apply Grover's algorithm to find the right information set. It leads to a quantum complexity of $\tilde{O} \left( 1/\sqrt{P_{prange}} \right)$.

\subsubsection{Generalized information set decoding.} The Prange algorithm has been improved in \cite{S88,D91} by relaxing a little bit the constraint on the set of columns we pick: it allows to have a little bit more than $0$ errors in the complementary of the information set $I$. To perform this task, the algorithm introduces two new parameters $p,l$ and looks for an error of the form $(\ev'|\ev'')$ where the right side has size $k+l$, $|\ev'| = w-p$, $|\ev''| = p$ with $\ev'$ uniquely determined by $\ev''$. More precisely, the improved algorithm first picks a set $I \subseteq \{1,\cdots,n\}$ of size $n-k-l$, then performs a Gaussian elimination on lines of $\Hm_{\pi_{I}}$ which gives a non-singular matrix $\Um$, as well as matrices $\Hm_{I}' \in \mathbb{F}_{2}^{(n-k-l)\times (k+l)}$ and $\Hm''_{I} \in \mathbb{F}_{2}^{l \times (k+l)}$ such that 
\begin{equation} 
\label{eq:1} 
\Um \Hm_{\pi_{I}} = \begin{pmatrix}
\mathbf{Id}_{n-k-l} & \Hm_{I}' \\
\mathbf{0} & \Hm_{I}'' \\ 
\end{pmatrix}
\end{equation} 
and 
\begin{equation} 
\label{eq:2} 
\Um \sv^{T} = (\sv_{I}'|\sv_{I}'')^{T} \mbox{ where } \sv_{I}' \in \mathbb{F}_{2}^{n-k-l} \mbox{ and } \sv_{I}'' \in \mathbb{F}_{2}^{l}.
\end{equation} 
Then if $\ev$ is a vector such that $\ev_{\pi_{I}} = (\ev'|\ev'')$ we have:
\begin{align*}
\Hm\ev^{T} = \sv^{T} &\iff \Um \Hm \ev^{T} = \Um \sv^{T} \\
&\iff  \begin{pmatrix}
\mathbf{Id}_{n-k-l} & \Hm' \\
\mathbf{0} & \Hm'' \\ 
\end{pmatrix} \ev_{\pi_{I}}^{T} = \begin{pmatrix} \sv'^{T} \\\sv''^{T} \end{pmatrix} \\
&\iff   \begin{pmatrix} \ev'^{T} + \Hm'\ev''^{T} \\ 
\Hm''\ev''^{T} \end{pmatrix} = \begin{pmatrix} \sv'^{T} \\\sv''^{T} \end{pmatrix} \\ 
&\iff \ev'^{T} = \Hm'\ev''^{T} + \sv'^{T} \mbox{ and } \Hm''\ev''^{T} = \sv''^{T} 
\end{align*}
In this way, we
compute all errors $\ev''$ of weight $p$ such that $\Hm''\ev'^{T} = \sv''^{T}$, for all vectors we get, we consider $\ev_{s} \eqdef (\ev''\Hm'^{T} + \sv' |\ev'')_{\pi_{I}^{-1}}$ and if one of them has a Hamming weight of $w$ then it is a solution, otherwise we pick another set of size $n-k-l$. Let us introduce now, for each subset $I$ we picked and syndrome $\sv$ we look to decode, the set: 
\begin{equation}
\label{eq:3}
\cS_{I} = \{ \ev'' \in \mathbb{F}_{2}^{k+l} \mbox{ of Hamming weight } p \mbox{ } : \mbox{ } \Hm_{I}'' \ev''^{T} = \sv_{I}''^{T} \}
\end{equation} 
\begin{equation}
\label{eq:4} 
f_{I} : \ev'' \in \mathbb{F}_{2}^{k+l} \mapsto \ev'' \Hm_{I}''^{T} \in  \mathbb{F}_{2}^{l} 
\end{equation} 
\begin{equation}
\label{eq:5} 
z_{I}^{\sv} : \ev'' \in \mathbb{F}_{2}^{k+l} \mapsto (\ev''\Hm_{I}'^{T}+\sv_{I}'|\ev'')_{\pi^{-1}} \in \mathbb{F}_{2}^{n}
\end{equation}
Thanks to equations \eqref{eq:1},\eqref{eq:2},\eqref{eq:3},\eqref{eq:4},\eqref{eq:5} we are able to formalize generalizes ISDs in Algorithm \ref{algo:isd}. 
\begin{algorithm}
	\caption{(generalized) ISD}\label{algo:isd}
	\begin{algorithmic}[1]
		\State {\bf input:} $\Hm \in \F_{2}^{(n-k)\times n}, \sv \in
		\F_{2}^{(n-k)}, l,p,w \mbox{ integers}$ 
		\Loop 
		\State pick a set $I \subseteq \{1,\cdots,n\}$ of size $n-k-l$ 
		\State compute $\Hm'_{I},\Hm''_{I}$, $\sv_{I}',\sv_{I}''$
		\State compute $\cS_{I}$
		\ForAll{$\ev''\in\cS_{I}$}
		\State
		$\ev \leftarrow h_{I}(\ev'')$ 
		\If{$\wt{\ev}=w$} {\bf output} $\ev$
		\EndIf
		\EndFor
		\EndLoop
	\end{algorithmic}
\end{algorithm}
\begin{remark} From each information set $I$ we can build matrices $\Hm_{I}'$, $\Hm_{I}''$, $\sv_{I}$ and $\sv_{I}''$ in polynomial time thanks to Gaussian elimination. 
\end{remark} 
This new algorithm leads to a probability $p_{p,l} \eqdef \frac{\binom{k+l}{p}\binom{n-k-l}{w-p}}{\binom{n}{w}}$ ($\geq p_{prange}$ for a set of parameters $p,l$) of finding a fixed solution. 
\begin{remark}We stress that to have this probability the algorithm has to consider all errors $\ev''$ of weight $p$ such that $\Hm'' \ev''^{T} = \sv''^{T}$.
\end{remark} 
\noindent In a same fashion as before this algorithm will succeed with probability:
$$
P_{p,l} \eqdef 1 - (1-p_{p,l})^{M_{k,n,w}} = \theta \big( \min \left( 1, M_{n,k,w} \cdot p_{p,l} \right) \big)  
$$
and if we denote by $T_{class}$ the time complexity to compute $S_{I}$, which is exponential as the size of $S_{I}$ is exponential, Algorithm \ref{algo:isd} has a complexity given by:
\[ \tilde{O}\left( \frac{T_{class}}{P_{p,l}} \right) \]

%Let us introduce now some definitions that will be useful in the following.

%\begin{definition}
%	 \label{def:isd} 
%	 Let parameters $p,l$, a parity-check matrix $\Hm \in \mathbb{F}_{2}^{(n-k)\times n}$ off full rank, a syndrome $\sv \in \mathbb{F}_{2}^{n-k}$ and $I \subseteq \{1,\cdots,n\}$ of size $n-k-l$. 
%	
%	
%	We define matrices $\Hm'_{I}$ and $\Hm''_{I}$ of size $(n-k-l)\times(k+l)$ and $l \times (k+l)$ such that it exists an invertible matrix $\Um$ of size $(n-k)\times(n-k)$ which gives:
%	$$
%	\Um \Hm_{\pi_{I}} =\begin{pmatrix}
%	\mathbf{Id}_{n-k-l} & \Hm_{I}' \\
%	\mathbf{0} & \Hm_{I}'' \\ 
%	\end{pmatrix}.
%	$$
%	syndromes $\sv_{I}',\sv_{I}''$ such that $\Um \sv^{T} = (\sv_{I}',\sv_{I}'')^{T}$. We then define functions: 
%	$$
%	f_{I} : \ev'' \in \mathbb{F}_{2}^{k+l} \mapsto \ev'' \Hm_{I}''^{T} \in  \mathbb{F}_{2}^{l} \quad ; \quad h_{I}^{\sv} : \ev'' \in \mathbb{F}_{2}^{k+l} \mapsto (\ev''\Hm_{I}'^{T}+\sv_{I}'|\ev'')_{\pi^{-1}} \in \mathbb{F}_{2}^{n}
%	$$
%	and the following set:
%	$$
%	\cS_{I} = \{ \ev'' \in \mathbb{F}_{2}^{k+l} \mbox{ of Hamming weight } p \mbox{ } : \mbox{ } \Hm_{I}'' \ev''^{T} = \sv_{I}''^{T} \}
%	$$
%	
%\end{definition}  
%
%
%We give in Algorithm \ref{algo:isd} the skeleton of ISDs with the above definition. 

Many classical algorithms have been proposed to solve Instruction 5 (see \cite{S88,D91,MMT11,BJMM12,MO15}). They all rely on splitting the matrices even more and finding elements $S_I$ via multi-collision algorithms. In the case of $\QDOOM$, similar ideas are applied. We generate several syndromes $\sv_1,\dots,\sv_q$. When performing the generalized ISD algorithm, we now have one set $S_I$ for each syndrome $\sv_q$. The multi-collision algorithms used in the ISD can take advantage of this in order to find all elements of all the $S_I$ (for different syndromes) in a reduced amortized cost. In this case, as we consider more good events, we obtain
$$
P_{p,l} = 1 - (1-p_{p,l})^{M_{k,n,w}} = \theta \big( \min \left( 1, q \cdot M_{n,k,w} \cdot p_{p,l} \right) \big)  
$$
Of course, in this case, the computing $T_{class}$ changes and new optimizations have to be done. We will not go into the details of these algorithms and optimizations (see \cite{S11} for more details).

The best asymptotic exponent among all those decoding techniques are \cite{MO15,BJMM12} for SD. However, algorithm \cite{MO15} is penalized by a big polynomial overhead which makes it more expensive that \cite{BJMM12}. It is why in the following table we will consider asymptotic exponents given by \cite{BJMM12}.  We give in Table \ref{table:classExp} classical exponents in base 2 of the Prange algorithm (which was the first algorithm proposed to solve syndrome decoding problem), \cite{BJMM12} and the state-of-the-art to solve {\DOOM}$_{\infty}$ (see \cite{S11}). We present the running times for $k = n/2$ and for two error weights $w$: namely $w \approx 0.11n$ which corresponds to the Gilbert-Varshamov weight and is the weight around which those problems are the hardest; and $w \approx 0.191n$ which corresponds to the weight used in the {\DST} signature scheme.

\setcounter{table}{0}

\begin{table}[H]
	\centering
	\begin{tabular}{|c||c|c|c|}\cline{2-4}
		\multicolumn{1}{c|}{} & \multicolumn{3}{|c|}{Classical asymptotic exponent in base $2$ (divided by $n$) } \\
		\hline
		$w/n$ & \qquad SD (Prange) \qquad \qquad  & \qquad SD (\cite{BJMM12}) \qquad \qquad  & $\QDOOM$ \cite{S11} \\
		\hline
		\hline 
		0.11 & 0.1199
		& 0.1000 & 0.0872 \\
		0.191 & 0.02029 & 0.01687 & 0.01654  \\
		\hline
	\end{tabular}
	\vspace{0.5cm}
	\caption{Asymptotic exponent for classically solving SD and {\DOOM}$_{\infty}$ for
		$k/n=0.5$}
\end{table}
\vspace*{-0.8cm}

The above table contains classical asymptotic exponent in base $2$ (divided by $n$). This means for example that the Prange algorithm for SD with $w = 0.11n$ runs in time $2^{0.1199n}$. \\

In the quantum setting, things become trickier. While Instruction 3 can be Groverized, it seems hard to get a full quadratic speedup for Instruction 5, because multi-collision problems have a less than quadratic speedup in the quantum setting. If $T_{quant}$ is the quantum running time of Instruction 5 then the total running time becomes $\tilde{O}\left( \frac{T_{quant}}{\sqrt{P_{p,l}}} \right)$. Moreover, any improvement we do in Instruction 5 seems to augment ${P_{p,l}}$ and therefore reduce the Grover advantage we have from Instruction 3. There seems to be very little place for improvement. In \cite{KT17}, authors still managed to find a quantum improvement over the simple quantum Prange algorithm using quantum random walks, even though the advantage is small.

\COMMENT{We can still perform Grover's algorithm on Instruction 3 but now the function for which we look for a root has an execution time $T_{quant}$ which is non-polynomial: Instruction 5 which consists in finding all errors of weight $p$ such that $\Hm''_{I}\ev''^{T} = \sv_{I}''^{T}$. Indeed for this problem, as it is shown by Proposition \ref{propo:numbS}, the typical number of solutions is $\binom{k+l}{p}/2^{l}$ which is exponential.
	It leads to a quantum complexity of $\tilde{O}\left( \frac{T_{quant}}{\sqrt{P_{p,l}}} \right) = \tilde{O} \left( \sqrt{\frac{T_{quant}^{2}}{P_{p,l}}} \right)$ while in classical case we have $\tilde{O}\left( \frac{T_{col}}{P_{p,l}} \right)$. The difference here between classical and quantum computation stands on the fact that we allow quantum algorithms to solve Instruction 5 and as it is shown in the previous formula, in order to halve the classical exponent of ISDs, we need to find a quantum algorithm which is twice efficient than in the classical case. Many classical algorithms have been proposed to solve Instruction 5 (see \cite{S88,D91,MMT11,BJMM12,MO15}). Nevertheless in any cases these algorithms compute solutions of $\Hm_{I}''\ev''^{T} = \sv_{I}''^{T}$ in amortized times $\tilde{O}(1)$ by using a lot of memories. Then as it was noticed in \cite{KT17}, there is a classical algorithm using less memory: Shamir-Shroeppel algorithm (see \cite{SS81}). Authors of \cite{KT17} then proposed to quantize this algorithm by using a random walk. We were inspired by their techniques to solve {\QDOOM}, it is the object of the following section.}

\subsection{Quantum Algorithm for solving {\QDOOM}}

We will focus on Instruction 5 and find the best tradeoffs for our quantum algorithm for $\QDOOM$. Similarly as in classical algorithms for SD, we will reduce our problem to a $k$-sum problem (actually a $4$-sum problem). Then by considering known results on quantum walks developed in \cite{KT17}, we will be able to give a running time for our quantum algorithm. Let us first introduce the following classical problem.

\begin{problem}[Generalized $k$-sum Problem]$ $
	\label{prob:sumP}

	Let $\cG$ be an Abelian group, $\cE$ be an arbitrary set, $k$ subsets $\cV_{1},\cdots,\cV_{k}$ of $\cE$, $k+1$ arbitrary maps:
	$$
	\forall i \in \llbracket 1,k \rrbracket, \mbox{ }  f_{i} : \cE \rightarrow \cG \quad ; \quad g : \cE^{k} \rightarrow \{0,1\}
	$$
	and an arbitrary $S \in \cG$. 	
	A solution is a tuple $(v_{1},\cdots,v_{k}) \in \cV_{1} \times \cdots \times \cV_{k}$ such that:
	\begin{itemize}
		\item $f_{1}(v_{1}) + \cdots + f_{k}(v_{k}) = S$ (subset-sum condition).

		\item $g(v_{1},\cdots,v_{k}) = 1$. 
	\end{itemize}
	
\end{problem}

We now show this reduction. Let $\Hm, \cH$ be an instance of {\QDOOM} and $\cH_{l}$ will denote the projection of $\cH$'s outputs onto their last $l$ coordinates.  We first pick an information set $I \subseteq \{1,\cdots,n\}$ of size $n-k-l$, then we build matrices $\Hm_{I}'$ and $\Hm_{I}''$ as in  \eqref{eq:1}.

\subsubsection{Associated $4$-sum problem.} We introduce the following sets and functions (see Equations \eqref{eq:3},\eqref{eq:4} and \eqref{eq:5}):
$$
\cG = \mathbb{F}_{2}^{l/2}\times \mathbb{F}_{2}^{l/2} \quad ; \quad \cE = \mathbb{F}_{2}^{k+l} ; S = 0
$$
$$
\forall i \in \llbracket 1,3 \rrbracket, \mbox{ }  f_{i} : \ev'' \in \mathbb{F}_{2}^{k+l} \mapsto \Hm_{I}'' \ev''^{T} ; f_{4} = \cH_{l}
$$
with
\begin{align*} 
\cV_{1} &\eqdef \{ (\ev_{1},\mathbf{0}_{2(k+l)/3}) \in \mathbb{F}_{2}^{k+l} \mbox{ } : \mbox{ } \ev_{1} \in \mathbb{F}_{2}^{(k+l)/3}, \mbox{ } |\ev_{1}| = p/3 \} \\
\cV_{2} &\eqdef \{ (\mathbf{0}_{(k+l)/3},\ev_{2},\mathbf{0}_{(k+l)/3}) \in \mathbb{F}_{2}^{k+l} \mbox{ } : \mbox{ }  \ev_{2} \in \mathbb{F}_{2}^{(k+l)/3}, \mbox{ } |\ev_{2}| = p/3 \} \\
\cV_{3} &\eqdef \{ (\mathbf{0}_{2(k+l)/3},\ev_{3}) \in \mathbb{F}_{2}^{k+l} \mbox{ } : \mbox{ }  \ev_{3} \in \mathbb{F}_{2}^{(k+l)/3}, \mbox{ } |\ev_{3}| = p/3 \} \\
\cV_{4} &\mbox{ be an arbitrary set of size } \binom{(k+l)/3}{p/3}
\end{align*} 
and
\[ g(v_{1},v_{2},v_{3},v_{4}) = 1 \iff |z^{\cH(v_{4})}_I(v_{1} + v_{2} + v_{3})| = w \]
\begin{proposition}
	If $(v_1,v_2,v_3,v_4)$ is a solution of the above problem then \\ $(v_1 + v_2 + v_3,v_4)$ is a solution of the $\QDOOM$ problem on inputs $(\Hm, \cH)$.
\end{proposition}
\begin{proof}
	Let $(v_1,v_2,v_3,v_4)$ a solution of the associated $4$-sum problem. We have 
	\begin{align*} 
	f_{1}(v_{1})+f_{2}(v_{2})+f_{3}(v_{3})+f_{4}(v_{4})=0 & \iff \Hm_{I}''(v_{1} + v_{2} + v_{3})^{T} = \cH_{l}(v_{4}) \\
	& \iff v_{1} + v_{2} + v_{3} \in \cS_{I} \ \mbox{ for the syndrome } \cH(v_{4}) 
	\end{align*} 
	This means that $|v_1 + v_2 + v_3| = p$. We also $g(v_1,v_2,v_3,v_4) = 1$ which implies $|z^{\cH(v_{4})}_I(v_{1} + v_{2} + v_{3})| = w$. By definition of $z_I$, this shows that 
	$$\Hm(z^{\cH(v_{4})}_I(v_{1} + v_{2} + v_{3})) = \mathcal{H}(v_4)$$
	which concludes the proof.
\end{proof}

All the above discussion was for a fixed information set $I$ so our goal is to use a quantum algorithm for the $4$-sum problem to solve instruction 5. Fortunately, there already exists a quantum study of this problem using quantum walks \cite[Proposition 2]{KT17}.

\begin{proposition} \label{Proposition:JPetGhazal} Consider the generalized $4$-sum problem defined in Problem \ref{prob:sumP} with sets $\cV_{i}$ of the same size $V$. Assume that $\cG$ can be decomposed as $\cG = \cG_{0} \times \cG_{1}$ with $|\cG_{0}|,|\cG_{1}| = \Theta(V^{4/5})$. There is a quantum algorithm (using a random walk) for solving the $4$-sum problem in running time $\tilde{O}\left( V^{6/5} \right)$.
\end{proposition} 

We now put everything together and present the running time of this quantum algorithm for $\QDOOM$.

\begin{theorem} We can solve {\QDOOM} for parameters $n,k$ and $w\geq d_{GV}(n,k)$ in time:
	$$
	\tilde{O}\left( \min_{0\leq l \leq n-k} \left( \frac{T_{quant}(p,l)}{\sqrt{P_{p,l}}}\right) \right) 
	$$
	where:
	$$
	P_{p,l} = \Theta \left( \min \left( 1, \frac{\binom{k+l}{p}\binom{n-k-l}{w-p}\binom{(k+l)/3}{p/3}}{2^{n-k}} \right) \right)  
	$$
	and 
	$$
	T_{quant}(p,l) = \binom{(k+l)/3}{p/3}^{6/5}
	$$
	with $p$ chosen such that:
	$$
	2^{l/2} = \Theta{\binom{(k+l)/3}{p/3}^{4/5}}
	$$	
\end{theorem}

The value of $T_{quant}$ is obtained from Proposition \ref{Proposition:JPetGhazal}. The other parameters are obtained from the classical analysis in the case where we consider $\binom{(k+l)/3}{p/3}$ syndromes. We present below quantum asymptotic exponents for SD and for $\QDOOM$. Again, we consider $k = n/2$ and for  error weights $w \approx 0.11n$ and $w \approx 0.191n$ which corresponds to the weight used in the {\DST} signature scheme.

\begin{table}[H]
	\centering
	\begin{tabular}{|c||c|c|c|}\cline{2-4}
		\multicolumn{1}{c}{}& \multicolumn{3}{|c|}{Quantum asymptotic exponent in base $2$ (divided by $n$)} \\ 
		\hline 
		$w/n$ &  SD (Prange)  & SD \cite{KT17}  &  \QDOOM (this work)  \\
		\hline
		\hline
		0.11 & 0.059958 & 0.058434
		& 0.056683
		\\
		0.191 & 0.010139 & 0.009218
		& 0.009159
		\\
		\hline
	\end{tabular}
	\vspace{0.3cm}
	\caption{Asymptotic exponent for quantumly solving SD and {\DOOM}$_{\infty}$ for
		$k/n=0.5$}
\end{table}

\section{Quantum security of the {\DST} signature scheme}
\label{sec:qsurf} 
We apply in this section our results to the SURF signature scheme presented in \cite{DST17}. Let us recall the condition upon which stands our security reduction in the {\QROM}: 
\cdt*
where $\lambda$ is the security parameters. Authors of \cite{DST17} proposed to use the family of $(U|U+V)$-codes as the secret key, namely:
\begin{definition}
	[$(U,U+V)$-Codes]
	Let $U$, $V$ be linear binary codes of length $n/2$ and dimension $k_{U}$, $k_{V}$. We define the subset of $\mathbb{F}_{2}^{n}$:
	\begin{displaymath}
	(U,U+V) \eqdef \{ (\uv,\uv+\vv) \mbox{ such that } \uv \in U \mbox{
		and } \vv \in V \}
	\end{displaymath}
	which is a linear code of length $n$ and dimension $k = k_{U} + k_{V}$. 
	%The resulting code is of minimum
	% distance $\min(2d_U,d_V)$ where $d_U$ is the minimum distance of $U$ and 
	% $d_V$ is the minimum distance of $V$.
\end{definition}

We choose parameters of public keys as:
$$
n = 13976 \ ; \ k = 6988 \ ; \ k_{U} = 4320 \ ; k_{V} = 2668 \ ; \ w = 2668.
$$

The value $n$ was chosen to get 128 bits of security for the $\QDOOM$ problem and the other parameters were already constrained (given $n$) from the specifications of SURF. We can now check the $3$ conditions.

\begin{enumerate}
	\item Using the results of \cite[Proposition 4]{DST17}, we get for our parameters \\ $\mathbb{E}_{\Hpub} \left( \rho(\cD_{w}^{\Hpub},\mathcal{U}_{n-k}) \right) = 2^{-0.06n}$ which gives if we choose a conservative $\qhash = 2^{128}$:
	$$
	\qhash^{\frac{3}{2}} \sqrt{\mathbb{E}_{\Hpub} \left( \rho(\cD_{w}^{\Hpub},\mathcal{U}_{n-k}) \right)} = \frac{1}{2^{235}}. 
	$$
	\item SURF performs a rejection sampling (see \cite[Section 5]{DST17}) algorithm that achieves $\rho(\mathcal{U}_w,\mathcal{D}_w) = 0$.
	\item While the authors of \cite{DST17} do not formally study quantum distinguishers for their code family, the best known classical algorithms not only also use multi-collision techniques and are hard even to Groverize. Also, for our parameters the classical advantage (see \cite[Section 7]{DST17}) is of the order of $2^{-500}$. Any quantum distinguisher for those codes would have to find radically new quantum algorithmic techniques way beyond the state of the art.
\end{enumerate}

Finally, with parameters and using the analysis of \cite{DST17}, we obtain the following parameters (we also include the parameters of the other quantum-safe signature schemes)

\begin{table}[H]
	\centering
	\caption{Security parameters for signature schemes with quantum security claims}
	\begin{tabular}{|c|c|c|c|c|}
		\hline\noalign{\smallskip}
		Scheme & Quantum security  & Public key size  & Private key size & Signature size \\
		& (in bits) & (in kBytes) & (in kBytes) & (in kBytes) \\
		\noalign{\smallskip}
		\hline
		\noalign{\smallskip}
		SPHINCS & $128$ & $1$  & $1$ & $41$ \\
		GPV-poly & 59 & $55$ & $26$ & $32$ \\
		GPV & 59 & $27840$ & $12064$ & $30$ \\
		TESLA-2 & $128$ & $21799$ & $7700$ & $4$ \\
		SURF & $128$ & $5960$ & $3170$ & $1.7$ \\
		\hline
	\end{tabular}
\end{table}

Moreover, for this choice of parameters the SURF signature scheme achieves a classical security of $231$ bits.

\section{Conclusion}
\label{sec:conclusion}
In this paper, we presented a method to perform tight security reductions for FDH-like signature schemes using code-based computational assumptions, more precisely on the $\QDOOM$ problem. We also analyzed the best known quantum algorithm for this problem. Finally, we applied our security reduction to the SURF signature scheme, presenting parameters for 128 bits of concrete quantum security and think this scheme will play an important role in the future standardization attempts from NIST. We finally list several open questions and perspectives that come out of this work:
\begin{itemize}
	\item Our security reduction can be applied to only one signature scheme now. Are there other constructions that could benefit from this reduction? The SURF signature scheme uses a code family which has very little structure. This strengthens the security but increases the key sizes. Can we use another code family that would stay secure with smaller key sizes?
	\item More generally, our techniques show that it is much better in the code-based setting to consider a computational assumption which starts from many instances of a problem and where we need to solve one of them. This One Out of Many approach appears implicitly when performing instance injection but doesn't appear explicitly in other signature schemes. For example, it would be very interesting to consider a One Out of Many equivalent for lattice schemes, and could be a way to reduce losses resulting from the quantum security reduction.
	\item Finally, since the security rely on the quantum hardness of the {\QDOOM} problem, it is important to continue to study it - similarly as other quantum-safe computational assumptions - in order to increase our trust in quantum secure schemes.
\end{itemize}

\addcontentsline{toc}{section}{Bibliography}
\newcommand{\etalchar}[1]{$^{#1}$}

\newpage
\appendix

\end{document}